\documentclass{scrartcl}

\usepackage{multirow}
\usepackage{mathrsfs}
\usepackage{authblk}
\usepackage{paralist}
\usepackage[pagebackref,menucolor=orange!40!black,filecolor=magenta!40!black,urlcolor=blue!40!black,linkcolor=red!40!black,citecolor=green!40!black,colorlinks]{hyperref}
\usepackage{mathtools}
\usepackage{amsthm}
\usepackage{subfigure} 
\usepackage{todonotes} 
\usepackage{caption}
\usepackage{amsmath}
\usepackage{amssymb}
\usepackage{comment}
\usepackage{amsfonts}
\makeatletter  
\newif\if@restonecol  
\makeatother

\usepackage[linesnumbered,ruled,vlined]{algorithm2e}
\usepackage{algpseudocode}  
\usepackage{amsmath}  

\usepackage{tikz}
\usetikzlibrary{positioning,backgrounds,patterns,calc}
\tikzstyle{vert}=[circle,draw=black,minimum size=8pt,inner sep=1pt]
\tikzstyle{vertex2}=[circle,draw=black,minimum size=15pt,inner sep=2pt]
\tikzstyle{edge}=[]
\tikzstyle{ypath}=[ultra thick]
\tikzstyle{dottedEdge}=[dotted,thick]
\tikzstyle{small-vertex}=[circle,draw=black,minimum size=6pt,inner sep=0pt,fill=white]
\tikzstyle{thinedges}=[draw=gray!30]
 \tikzstyle{boxes}=[draw,thick, rounded corners=3mm,text width=2.7cm,align=center,text opacity=1,fill opacity=1,fill=white]
\tikzstyle{unk}=[fill=gray!25!white]
\usetikzlibrary{decorations.pathreplacing}

\usepackage{bbding}
\usepackage{placeins}
\usepackage{dsfont}

\newif\ifautorefabbr
\autorefabbrfalse

\theoremstyle{theorem}
\newtheorem{thm}{Theorem}{\bfseries}{\normalfont}
\newtheorem{obs}{Observation}{\bfseries}{\normalfont}
\newtheorem{cor}{Corollary}{\bfseries}{\normalfont}
{\bfseries}{\normalfont}
\newtheorem{lem}{Lemma}{\bfseries}{\normalfont}
\theoremstyle{definition}
\newtheorem{rrule}{Reduction Rule}{\bfseries}{\normalfont}
\newtheorem{rrulev}{Reduction Rule}{\bfseries}{\normalfont}
\newtheorem{defi}{Definition}{\bfseries}{\normalfont}

\DeclareMathOperator{\dist}{dist}

\usepackage{chngcntr}

\newcounter{rrulevariant}
\counterwithin{rrulev}{rrulevariant}

\newcommand{\rrulecounter}{\refstepcounter{rrulevariant}}

\newcommand{\decprob}[3]{%
  \begin{center}%
    \begin{minipage}{0.9\linewidth}%
      \textsc{#1}\\
      \textbf{Input:} #2\\
      \textbf{Question:} #3
    \end{minipage}%
  \end{center}%
}

\newcommand{\decprobnolineend}[3]{%
  \begin{center}%
    \begin{minipage}{0.9\linewidth}%
      \textsc{#1}\\
      \textbf{Input:} #2
      \textbf{Question:} #3
    \end{minipage}%
  \end{center}%
}

\newcommand{\CE}{\textsc{Cluster Editing}\xspace}

\newcommand{\CC}{\textsc{Cluster Completion}\xspace}
\newcommand{\ODCE}{\textsc{Original Dynamic Cluster Editing}\xspace}
\newcommand{\DCE}{\textsc{Dynamic Cluster Editing}\xspace}
\newcommand{\DCC}{\textsc{Dynamic Cluster Completion}\xspace}
\newcommand{\DCD}{\textsc{Dynamic Cluster Deletion}\xspace}

\newcommand{\DCEEDlong}{\textsc{Dynamic Cluster Editing with Edge-Based Distance}\xspace}
\newcommand{\DCDEDlong}{\textsc{Dynamic Cluster Deletion with Edge-Based Distance}\xspace}
\newcommand{\DCCEDlong}{\textsc{Dynamic Cluster Completion with Edge-Based Distance}\xspace}
\newcommand{\DCshort}{DC}
\newcommand{\DCEED}{\textsc{\DCshort Editing (Edge Dist)}\xspace}
\newcommand{\DCDED}{\textsc{\DCshort Deletion (Edge Dist)}\xspace}
\newcommand{\DCCED}{\textsc{\DCshort Completion (Edge Dist)}\xspace}
\newcommand{\DCEMDlong}{\textsc{Dynamic Cluster Editing with Matching-Based Distance}\xspace}
\newcommand{\DCDMDlong}{\textsc{Dynamic Cluster Deletion with Matching-Based Distance}\xspace}
\newcommand{\DCCMDlong}{\textsc{Dynamic Cluster Completion with Matching-Based Distance}\xspace}
\newcommand{\DCEMD}{\textsc{\DCshort Editing (Matching Dist)}\xspace}
\newcommand{\DCDMD}{\textsc{\DCshort Deletion (Matching Dist)}\xspace}
\newcommand{\DCCMD}{\textsc{\DCshort Completion (Matching Dist)}\xspace}
\newcommand{\MCK}{\textsc{Multi-Choice Knapsack}\xspace}
\newcommand{\XC}{\textsc{Exact Cover by 3-Sets}\xspace}

\usepackage[numbers,sort]{natbib}
\makeatletter
\def\NAT@spacechar{~}
\makeatother

\newcommand{\Clique}{\textsc{Clique}\xspace}
\newcommand{\CliqueReg}{\textsc{Clique on Regular Graphs}\xspace}
\newcommand{\MCC}{\textsc{Multicolored Clique}\xspace}

\newcommand{\Partition}{\textsc{3-Partition}\xspace}

\newcommand{\NP}{\ensuremath{\textsf{NP}}\xspace}

\newcommand{\FPT}{\ensuremath{\textsf{FPT}}\xspace}
\newcommand{\W}[1]{\ensuremath{\textsf{W[#1]}}\xspace}

\newcommand{\N}{\mathds{N}}

\graphicspath{{images/}}

\usepackage{etoolbox}
\usepackage{booktabs}

\newcommand{\appsymb}{$\star$}
\newcommand{\appref}[1]{\appsymb}
\newcommand{\apprefProof}[1]{\appsymb}

\usepackage[noabbrev,nameinlink,capitalize]{cleveref}
\crefname{table}{Table}{Tables}
\crefname{figure}{Figure}{Figures}
\crefname{cor}{Corollary}{Corollaries}
\crefname{step}{Step}{Steps}
\crefname{rrule}{Reduction Rule}{Reduction Rules}
\crefname{rrulev}{Reduction Rule}{Reduction Rules}
\crefname{thm}{Theorem}{Theorems}
\Crefname{thm}{Thm.}{Thm.}
\crefname{obs}{Observation}{Observations}
\crefname{lem}{Lemma}{Lemmas}
\crefname{section}{Section}{Sections}

\title{Parameterized Dynamic Cluster Editing\footnote{An extended abstract of this work appears in the proceedings of the 38th IARCS Annual Conference on Foundations of Software Technology and Theoretical Computer Science (FSTTCS~'18)~\cite{LMNN18}.
Unfortunately, the conference version contains a claim whose unpublished proof
contained an error (see~\cref{table:main-results}). This full version contains
all proof details.}}

\author[1]{Junjie~Luo\thanks{Supported by
CAS-DAAD Joint Fellowship Program for Doctoral Students of UCAS.
Work done while with TU Berlin.}}
\affil[1]{{\small Algorithmics and Computational Complexity, Faculty~IV, TU Berlin, {Berlin, Germany}\\
Academy of Mathematics and Systems Science, Chinese Academy of Sciences, Beijing, China\\ School of Mathematical Sciences, University of Chinese Academy of Sciences, Beijing, China\\ \texttt{junjie.luo@campus.tu-berlin.de, luojunjie@amss.ac.cn}}}

\author[2]{Hendrik~Molter\thanks{Supported by the DFG, project MATE (NI 369/17).}}
\affil[2]{{\small Algorithmics and Computational Complexity, Faculty~IV, TU Berlin, {Berlin, Germany}\\ \texttt{\{h.molter, andre.nichterlein, rolf.niedermeier\}@tu-berlin.de}}}

\author[2]{Andr\'{e}~Nichterlein}

\author[2]{Rolf~Niedermeier}

\date{}

\begin{document}

\maketitle

\begin{abstract}
We introduce a dynamic version of the \NP-hard graph problem \textsc{Cluster Editing}.
The essential point here is to take into account dynamically evolving input 
graphs: Having a cluster graph (that is, a disjoint union of cliques) regarding a solution for thr first input graph, can we cost-efficiently transform it into
a ``similar'' cluster graph that is a solution for the second (``subsequent'') input graph? 
This model is motivated by several application scenarios, including 
incremental clustering, the search for compromise clusterings, or also local 
search in graph-based data clustering.
We thoroughly study six problem variants (edge editing, edge deletion, 
edge insertion; each combined with two distance measures between cluster graphs).
We obtain both fixed-parameter tractability as well as (parameterized) hardness
results, thus (except for three open questions) providing a fairly complete
picture of the parameterized computational complexity landscape under the two 
perhaps most natural parameterizations: the distances of the new ``similar'' cluster graph to (i)~the second input graph 
and to (ii)~the input cluster graph.

\bigskip

\end{abstract}

\section{Introduction}\label{sec:intro}
The \NP-hard \CE{} problem~\cite{ben-dor_clustering_1999,SST04}, also known as 
\textsc{Correlation Clustering}~\cite{bansal2004correlation}, is one of the most popular graph-based data clustering problems 
in algorithmics. 
Given an undirected graph, the task is to transform it into a disjoint 
union of cliques (also known as cluster graph) 
by performing a minimum number of edge modifications 
(deletions or insertions). Being \NP-hard, \CE{} gained 
high popularity in studies concerning parameterized 
algorithmics, e.g.~\cite{Abu17,AEGS018,bocker_cluster_2013,cao_cluster_2012,chen2017parameterized,fomin_tight_2014,GGHN05,HH15,komusiewicz2012cluster}. To the best of our knowledge, 
to date these parameterized studies mostly focus on a ``static 
scenario''. 
\citet{chen2017parameterized} are an exception by also
looking at temporal and multilayer graphs. In their work, the input is a set of
graphs (multilayer) or an ordered list of graphs (temporal), in both cases
defined over the same vertex set. The goal is to transform each input graph
into a cluster graph
 such that, in the multilayer case, the number of vertices in which any two
 cluster graphs may differ is upper-bounded, and in the temporal case, the
 number of vertices in which any two consecutive (with respect to their position
 in the list) cluster graphs may differ is upper-bounded. 
 
In this work, we introduce a dynamic view on \CE by, roughly speaking, assuming
that the input graph changes once. 
Thus we seek to efficiently and effectively adapt an existing
solution, namely a cluster graph. In contrast to the work of
\citet{chen2017parameterized}, we do \emph{not} assume that all future changes are known. We consider the scenario where given an input graph, we only know changes that lie immediately ahead, that is, we know the ``new'' graph that the input graph changes to.
Motivated by the assumption that 
the ``new'' cluster graph should only change moderately but still be a valid representation of the data, 
we parameterize both on the 
number of edits necessary to obtain the ``new'' cluster graph 
and the difference between the ``old'' and the ``new'' cluster graph.
We finally remark that there have been previous parameterized studies of 
dynamic (or incremental) graph problems, dealing with 
coloring~\cite{HN13}, domination~\cite{Downey+2014,Abu-KhzamCESW17}, or 
vertex deletion~\cite{Abu-Khzam+2015,krithika2018dynamic}
problems. 

\subparagraph{Mathematical model.}
In principle, the input for a dynamic version of a static problem~$X$ are two instances~$I$ and~$I'$ of~$X$, a solution~$S$ for~$I$, and an integer~$d$.
The task is to find a solution~$S'$ for~$I'$ such that the distance between~$S$ and~$S'$ is upper-bounded by~$d$.
Often, there is an additional constraint on the size of~$S'$.
Moreover, the symmetric difference between~$I$ and~$I'$ is used as a parameter for the problem many times.
We arrive at our following ``original dynamic version'' of \CE{} (phrased as
decision version).

\decprobnolineend{\ODCE}
{Two undirected graphs~$G_1$ and~$G_2$ and a cluster graph $G_c$ over the same
vertex set, and two nonnegative integers: a budget~$k$ and a distance upper
bound~$d$ such that
$|E(G_1) \oplus E(G_c)| \le k$. 

}
{Is there a cluster graph~$G'$ for~$G_2$ such that
\begin{compactenum}
\item $|E(G_2) \oplus E(G')| \le k$ and
\item $\dist(G',G_c) \le d$?
\end{compactenum}
}

Herein, $\oplus$ denotes the symmetric difference between two sets and $\dist(\cdot,\cdot)$ is a generic distance function for cluster graphs, which we discuss later. 
Moreover, $G_c$ is supposed to be the ``solution'' given for the input graph~$G_1$.
However, since the question in this problem formulation is independent from $G_1$ we can remove this graph from the input and arrive at the following simplified version of the problem.
For the remainder of this paper we focus on this simplified formulation of \DCE.

\decprob{\DCE}
{An undirected graph~$G$ and a cluster graph~$G_c$ over the same vertex set, and
two nonnegative integers: a budget~$k$ and a distance upper bound~$d$.} {Is
there a cluster graph~$G'$ for~$G$ such that
 \begin{compactenum}
  \item $|E(G) \oplus E(G')| \le k$ and
  \item $\dist(G', G_c) \le d$?
\end{compactenum}
}

There are many different distance measures for cluster graphs~\cite{meila2005comparing,meila2012local}.
Indeed, we will study two standard ways of measuring the distance between two 
cluster graphs. 
One is called classification error distance, which measures the number of vertices one needs to move between cliques to make two cluster graphs the same---we subsequently refer to it as \emph{matching-based distance}. 
The other is called disagreement distance, which is the 
symmetric distance between two edge sets---we subsequently refer to it as \emph{edge-based distance}. 
Notably, the edge-based distance upper-bounds the matching-based distance.
We give formal definitions in \cref{sec:prelims}.

\subparagraph{Motivation and related work.}
Beyond parameterized algorithmics and static \CE{}, 
dynamic clustering in general has been subject to 
many studies, mostly in applied computer 
science~\cite{tang_clustering_2009,dong_clustering_2012,dey_temporal_2017,tantipathananandh_framework_2007,tantipathananandh_finding_2011,charikar2004incremental}. 
We mention in passing that there are also close ties to reoptimization 
(e.g., \cite{BHMW08,schieber2018theory,abs-1809-10578}) and parameterized local
search (e.g., \cite{FFLRSV12,GKOSS12,GHNS13,HN13,MS10}).  

There are several natural application scenarios that motivate the 
study of \DCE{}. Next, we list four of them.

\begin{description}
\item [\textbf{Dynamically updating an existing cluster graph.}] \DCE{} can be interpreted to model a smooth transition 
between cluster graphs, reflecting that ``customers'' working 
with clustered data in a dynamic setting may only tolerate a moderate 
change of the clustering from ``one day to another'' since ``revolutionary'' 
transformations would require too dramatic changes in their work.
In this spirit, when employing small parameter values, \DCE{} has kind of an evolutionary flavor with respect to the history of the various cluster graphs in a dynamic setting.
\item [\textbf{Editing a graph into a target cluster graph.}] For a given graph~$G$,
there may be many cluster graphs which are at most~$k$ edge modifications away.
The goal then is to find one of these which is close to the given target
cluster graph~$G_c$ since in a corresponding application one is already ``used to'' work with~$G_c$. 
Adapting a different point of view, the editing into the target cluster
graph~$G_c$ might be too expensive (that is, $|E(G) \oplus E(G_c)|$ is too big),
and one has to find a solution cluster graph with small enough modification
costs but being still close to the target~$G_c$.
\item[\textbf{Local search for an improved cluster graph.}]
Here the scenario is that one may have found an initial clustering expressed 
by~$G_c$, and one searches for another solution~$G'$ for~$G$ within a certain local region around~$G_c$ (captured by our parameter~$d$).
\item[\textbf{Editing into a compromise clustering.}]
When focusing on the edge-based distance, one may generalize the definition 
of \DCE{} by allowing $G_c$ to be any graph (not necessarily a cluster graph).
This may be used as a model for ``compromise cluster editing''
in the sense that the goal cluster graph then is a compromise for 
a cluster graph suitable for both input graphs since it is close to both of them.
\end{description}

\newcommand{\mrrb}[2]{\multirow{#1}{*}{\rotatebox[origin=c]{90}{#2}}}
\begin{table}[t]
\caption{Result overview for \DCE{}. We primarily categorize the problem variants by the distance measure (Matching, Edge) they use and secondarily by the allowed modification operation. 
\NP-completeness for all problem variants (last column) even holds if the input graph~$G$ is a cluster graph. PK stands for polynomial problem kernel.}
\label{table:main-results}
\centering
\begin{tabular}{l@{\hspace{.5em}}l|l@{\hspace{.6em}}l|l@{\hspace{.4em}}l|l@{\hspace{.4em}}l|l@{\hspace{.7em}}l}

\toprule
&&\multicolumn{6}{c|}{Parameter}&&\\
\multicolumn{2}{l|}{Problem Variant} & \multicolumn{2}{l|}{$k+d$} & \multicolumn{2}{l|}{$k$} & \multicolumn{2}{l|}{$d$} &  & \\ 
\midrule
\mrrb{3}{\small Matching} & Editing & \FPT(PK) & \mrrb{3}{\Cref{thm:polykernel}}  & \W1-h &
\Cref{thm:Whard} & \W1-h  & \multirow{2}{*}{$\Big\}$ \Cref{thm:Whard}} &\NP-c & \mrrb{3}{\Cref{thm:completionhardness}}\\
& Deletion & \FPT(PK) & & \emph{open} & & \W1-h & &\NP-c &\\
& Completion & \FPT(PK) & & \emph{open} & & \FPT  & \Cref{thm:fpt} & \NP-c & \\
\midrule
\mrrb{3}{\small Edge} & Editing & \FPT(PK) & \mrrb{3}{\Cref{thm:polykernel}} & \W1-h & \Cref{thm:Whard} & \W1-h & \multirow{2}{*}{$\Big\}$ \Cref{thm:Whard}} &\NP-c & \mrrb{3}{\Cref{thm:completionhardness}}\\
& Deletion & \FPT(PK) & & \FPT & 
\Cref{thm:fpt}
& \W1-h & &\NP-c &\\
& Completion & \FPT(PK) & & \emph{open}\footnotemark & & \FPT &
\Cref{thm:fpt} & \NP-c & \\
\bottomrule
\end{tabular}
\end{table}
\footnotetext{In the conference version~\cite{LMNN18} of this paper we claimed
that \DCCED is in FPT when parameterized by~$k$. Unfortunately, the unpublished
proof for this claim contained an irreparable error.}

\subparagraph{Our results.}
We investigate the (parameterized) computational complexity of
\DCE{}. We study \DCE{} as well as two restricted versions where 
only edge deletions (``Deletion'') or edge insertions (``Completion'') 
are allowed. We show that all problem variants (notably also the completion variants, whose static counterpart is trivially polynomial-time solvable) are \NP-complete even if the input graph~$G$ is already a cluster graph. 
\Cref{table:main-results} surveys our main complexity results.

The general versions of \DCE{} all turn out to be parameterized intractable
(\W1-hard) by the single natural parameters ``budget~$k$'' and
``distance~$d$''; however, when both parameters are combined, one achieves a
polynomial-size problem kernel, also implying fixed-parameter tractability.
We also derive a generic approach, based on a reduction to \MCK, to derive
fixed-parameter tractability for several deletion and completion variants with
respect to the parameters budget~$k$ as well as the
distance~$d$.

\subparagraph{Organization of the paper.}
Our work, after introducing basic notation (\cref{sec:prelims}), consists of two main parts.
In \cref{sec:hardness}, we provide all our (parameterized) hardness results.
In \cref{sec:fpt}, we develop several positive algorithmic results, namely
polynomial-size problem kernels through polynomial-time data reduction, and
fixed-parameter solving algorithms.
We conclude with a summary and directions for future work (\cref{sec:concl}).

\section{Preliminaries and Problems Variants} 
\label{sec:prelims}
In this section we give a brief overview on concepts and notation of graph theory and parameterized complexity theory that are used in this paper. We also give formal definitions of the distance measures for cluster graphs we use and of our problem variants.
We use~$\oplus$ to denote the symmetric difference, that is, for two sets~$A,B$
we have~$A \oplus B := (A \setminus B) \cup (B \setminus A)$.

\subparagraph{Graph-theoretic concepts and notations.} 
Given a graph $G=(V,E)$, we say that a vertex set $C\subseteq V$ is a \emph{clique in~$G$} if $G[C]$ is a complete graph. 
We say that a vertex set $C\subseteq V$ is \emph{isolated} in $G$ if there is no edge~$\{u,v\}\in E$ with $u\in C$ and $v\in V\setminus C$. 
A $P_3$ is a path with three vertices. 
We say that vertices $u,v,w\in V$ form an induced $P_3$ in $G$ if $G[\{u,v,w\}]$ is a $P_3$. 
We say that an edge~$\{u,v\}\in E$ is part of a $P_3$ in $G$ if there is a vertex~$w\in V$ such that $G[\{u,v,w\}]$ is a $P_3$. 
Analogously, we say that a non-edge~$\{u,v\}\notin E$ is part of a $P_3$ in $G$ if there is a vertex~$w\in V$ such that $G[\{u,v,w\}]$ is a $P_3$.
A graph $G=(V,E)$ is a \emph{cluster graph} if for all $u,v,w\in V$ we have that $G[\{u,v,w\}]$ is not a $P_3$, or in other words, $P_3$ is a forbidden induced subgraph for cluster graphs.

\subparagraph{Distance measures for cluster graphs.}
A cluster graph is simply a disjoint union of cliques.
We use two basic distance measures for cluster
graphs~\cite{meila2005comparing,meila2012local}. The first one is called ``matching-based
distance'' and counts how many vertices have to be moved from one cluster to another to make two cluster graphs the same.  See \cref{fig: d_M
and d_E} for an illustrating example. It is formally defined as follows.
\begin{defi}[Matching-based distance]
Let $G_1=(V,E_1)$ and $G_2=(V,E_2)$ be two cluster graphs defined over the same vertex set. 
Let~$B(G_1,G_2)=(V_1\uplus V_2, E, w)$ be a weighted complete bipartite graph, where each vertex $u\in V_1$ corresponds to one cluster of~$G_1$, denoted by $C_u\subseteq V$, and each vertex $v\in V_2$ corresponds to one cluster of~$G_2$, denoted by $C_v\subseteq V$. 
The weight of the edge between $u\in V_1$ and $v\in V_2$ is~$w(\{u,v\})=|C_u\cap C_v|$. 
Let~$W$ be the weight of a maximum-weight matching in $B(G_1,G_2)$.
The \emph{matching-based distance}~$d_M$ between~$G_1$ and~$G_2$ is $d_M(G_1,G_2):=|V|-W$.
\end{defi}
The second distance measure is called ``edge-based distance'' and simply measures the symmetric distance between the edge sets of two cluster graphs.
\begin{defi}[Edge-based distance]
Let $G_1=(V,E_1)$ and $G_2=(V,E_2)$ be two cluster graphs defined over the same vertex set. The \emph{edge-based distance} $d_E$ between~$G_1$ and~$G_2$ is~$d_E(G_1,G_2):=|E_1 \oplus E_2|$.
\end{defi}
See \cref{fig: d_M and d_E} for an example illustration of two cluster graphs~$G_1$ and~$G_2$ defined over the same vertex set
$V=\{u_1,u_2,u_3,u_4,u_5,u_6,v_1,v_2,w\}$.
In~$G_1$ there are three cliques (clusters)~$C_1=\{u_1,u_2,u_3,u_4,u_5,u_6\}$, $C_2=\{v_1,v_2\}$ and~$C_3=\{w\}$.
In~$G_2$ there are two cliques~${C_1}'=\{u_1,u_2,u_3,v_1,v_2\}$ and~${C_2}'=\{u_4,u_5,u_6,w\}$.
Then in~$B(G_1,G_2)$ we have three vertices on the left side for the cliques in~$G_1$ and two vertices on the right side for the cliques in~$G_2$.
A maximum-weight matching for~$B(G_1,G_2)$ matches~${C_1}$ with~$C_2'$ and~${C_2}$ with~$C_1'$, and has weight~$W=5$.
Thus we have~$d_M(G_1, G_2)=|V|-W=9-5=4$, 
while~$d_E(G_1,G_2)=3^2 +2 \cdot 3 +1 \cdot 3 =18$.

\begin{figure}[t]
\begin{center}
\begin{tikzpicture}[line width=1pt, scale=1]
		 	
\node[draw,circle,inner sep=1pt] (u1) at (1,0) {$u_1$};
\node[draw,circle,inner sep=1pt] (u2) at (2,0) {$u_2$};
\node[draw,circle,inner sep=1pt] (u3) at (3,0) {$u_3$};
\node[draw,circle,inner sep=1pt] (u4) at (4,0) {$u_4$};
\node[draw,circle,inner sep=1pt] (u5) at (5,0) {$u_5$};
\node[draw,circle,inner sep=1pt] (u6) at (6,0) {$u_6$};
\node[draw,circle,inner sep=1pt] (v1) at (1.5,-1) {$v_1$};
\node[draw,circle,inner sep=1pt] (v2) at (2.5,-1) {$v_2$};
\node[draw,circle,inner sep=1pt] (w) at (5,-1) {$w$};

\draw[blue, rounded corners=15pt,dashed]  (0.6,-1.6) rectangle ++(2.8,2.2);	
\draw[blue, rounded corners=15pt,dashed]  (3.6,-1.6) rectangle ++(2.8,2.2);	
\draw[red, rounded corners=10pt,dotted]  (0.4,-0.4) rectangle ++(6.2,0.8);	
\draw[red, rounded corners=10pt,dotted]  (0.9,-1.4) rectangle ++(2.2,0.8);	
\draw[red, dotted] (5,-1) circle (0.4cm);	

\node[blue]  at (2,-2) {${C_1}'$};	
\node[blue]  at (5,-2) {${C_2}'$};	
\node[red]   at (0,0) {$C_1$};
\node[red]   at (0,-1) {$C_2$};
\node[red]   at (4.2,-1) {$C_3$};	

\begin{scope}[xshift=2.5cm, yshift=0cm]
\node[blue, draw,circle, scale=1] (C1') at (10,0.5) {};
\node[blue] () at (10.5,0.5) {${C_1}'$};
\node[blue, draw,circle, scale=1] (C2') at (10,-1.5) {};
\node[blue] () at (10.5,-1.5) {${C_2}'$};

\node[red, draw,circle, scale=1] (C1) at (8,0.5) {};
\node[red] () at (7.5,0.5) {$C_1$};
\node[red, draw,circle, scale=1] (C2) at (8,-0.5) {};
\node[red] () at (7.5,-0.5) {$C_2$};
\node[red, draw,circle, scale=1] (C3) at (8,-1.5) {};
\node[red] () at (7.5,-1.5) {$C_3$};
\draw (C1') -- (C1) node [midway,above] {3};
\draw[line width=0.7mm] (C1') -- (C2) node [midway,above] {\textbf{2}};
\draw (C1') -- (C3) node [near start,right] {0};
\draw[line width=0.7mm] (C2') -- (C1) node [near start,above] {\textbf{3}};
\draw (C2') -- (C2) node [midway,below] {0};
\draw (C2') -- (C3) node [midway,below] {1};
\node () at (6.2,0) {$B(G_1,G_2)\colon$};
\end{scope}
	
\end{tikzpicture}
\caption{An illustration of the matching-based distance measure.
On the left side, red dotted boundaries represent cliques in cluster graph~$G_1$, and blue dashed boundaries represent cliques in cluster graph~$G_2$.
The bipartite graph on the right side is the edge-weighted bipartite graph~$B(G_1,G_2)$.
The maximum-weight matching for~$B(G_1,G_2)$ is formed by the two edges represented by the two bold lines.}
\label{fig: d_M and d_E}
\end{center}
\end{figure}
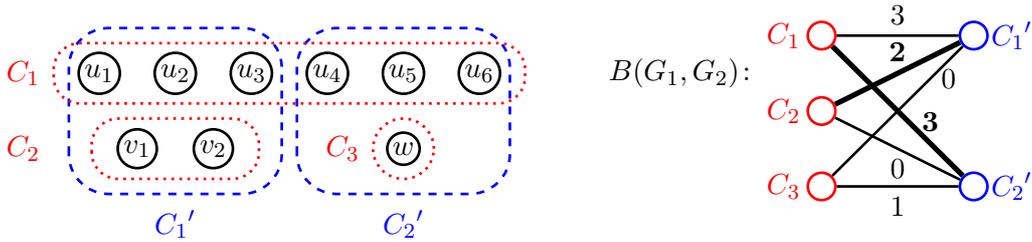

\subparagraph{Problem names and definitions.} 
In the following we present the six problem variants we are considering. We use \DCE\ as a basis for our problem variants. 
In \DCD\ we add the constraint that $E(G') \subseteq E(G)$ and in \DCC\ we add
the constraint that $E(G) \subseteq E(G')$.
For each of these three variants we distinguish a matching-based version and an edge-based version, where the generic ``dist'' in the problem definition of \DCE is replaced by~$d_M$ and~$d_E$, respectively. 
This gives us a total of six problem variants. 
We use the following abbreviations for our problem names. The letters ``DC'' stand for ``Dynamic Cluster'', and ``Matching Dist'' is short for ``Matching-Based Distance''. 
Analogously, ``Edge Dist'' is short for ``Edge-Based Distance''. 
This yields the following list of studied problems:
\begin{compactitem}
	\item \DCEMDlong, \\ abbreviation: \DCEMD. 
	\item \DCDMDlong, \\ abbreviation: \DCDMD. 
	\item \DCCMDlong, \\ abbreviation: \DCCMD. 
	\item \DCEEDlong, \\ abbreviation: \DCEED. 
	\item \DCDEDlong, \\ abbreviation: \DCDED. 
	\item \DCCEDlong, \\ abbreviation: \DCCED. 
\end{compactitem}

\subparagraph{Parameterized complexity.} 
We use standard notation and terminology from parameterized
complexity~\cite{downey2013fundamentals,flum2006parameterized,Nie06,CyganFKLMPPS15}
and give here a brief overview of the most important concepts.
A \emph{parameterized problem} is a language $L\subseteq \Sigma^* \times \mathbb{N}$, where $\Sigma$ is a finite alphabet. We call the second component
the \emph{parameter} of the problem.
A parameterized problem is \emph{fixed-pa\-ram\-e\-ter tractable} (in the complexity class \FPT{})
if there is an algorithm that solves each instance~$(I, r)$ in~$f(r) \cdot |I|^{O(1)}$ time,
for some computable function $f$. 
A parameterized problem $L$ admits a \emph{polynomial kernel} if there is a polynomial-time algorithm that transforms each instance $(I,r)$ into an instance $(I', r')$ such that $(I,r)\in L$ if and only if $(I',r')\in L$ and $|(I', r')|\le f(r)$,
for some computable function $f$. 
If a parameterized problem is hard for the parameterized complexity class \W1, then it is (presumably) not in~\FPT{}.
The complexity class \W1 is closed under parameterized reductions, which may run in \FPT-time and additionally set the new parameter to a value that exclusively depends on the old parameter.

\section{Intractability Results}\label{sec:hardness}
In this section we first establish \NP-completeness for all problem variants of \DCE, even if the input graph~$G$ is already a cluster graph. 
\begin{thm}
	\label{thm:completionhardness}
	All considered problem variants of \DCE\ are \NP-complete, even if the input graph~$G$ is a cluster graph.
\end{thm}
Intuitively, \cref{thm:completionhardness} means that on top of the \NP-hard
task of transforming a graph into a cluster graph, it is computationally hard to improve an already found clustering with respect to being closer to the target cluster graph.
Notably, while the dynamic versions of \CC\ turn out to be \NP-complete, it is
easy to see that classical \CC is solvable in polynomial time.

In a second part of this section we show \W1-hardness results both for budget
parameter~$k$ and for distance parameter~$d$ for several variants of \DCE.
Formally, we show the following.

\begin{thm}
	\label{thm:Whard}
	The following problems are \W1-hard when parameterized by the budget~$k$:
	\begin{compactitem}
		\item \DCEMD,
		\item \DCEED.
	\end{compactitem}
	The following problems are \W1-hard when parameterized by the distance~$d$:

	\begin{compactitem}
		\item \DCEMD,
		\item \DCDMD,
		\item \DCEED, and
		\item \DCDED.
	\end{compactitem}
\end{thm}

The proof of \cref{thm:Whard} is based on several parameterized reductions which are presented in \cref{sec:Whardness}.
The proof of \cref{thm:completionhardness} is based on nonparameterized polynomial-time many-one reductions (see \cref{sec:NP-hardness}) and some parameterized reductions that also imply NP-hardness (see \cref{sec:Whardness}).
More precisely, \cref{thm:completionhardness} follows from
\cref{lem:DCCEDhard,lem:DCCMDhard,obs: symmetric swap,lem:DCDMDhard} presented
in 
\cref{sec:NP-hardness}, as well as \cref{lem:DCEMDWhardwrtk,lem:DCEEDWhardwrtk}
presented in \cref{sec:Whardness}.

\subsection{Polynomial-time many-one reductions}
\label{sec:NP-hardness}
 
We first present two polynomial-time many-one reductions from the strongly \NP-hard \Partition problem~\cite{GJ79} for both \DCCMD and \DCCED with input graphs~$G$ that are already cluster graphs. We start with the latter. 
\begin{lem}
\label{lem:DCCEDhard}
\DCCED is \NP-complete, even if the input graph~$G$ is a cluster graph.
\end{lem}
\begin{proof}
We present a polynomial-time reduction from \Partition, where given a multi-set
of~$3m$ positive integers~$\{a_1,a_2,\dots,a_{3m}\}$ with~$\sum_{1 \le i \le
3m}a_i=mB$ and for~$1 \le i \le 3m$ it holds that $B/4 < a_i < B/2$, the task is
to determine whether this multi-set can be partitioned into~$m$~disjoint
subsets~$A_1,A_2, \dots, A_m$ such that for each~$1 \le i \le m$, $\sum_{a_j
\in A_i}a_j=B$.
Given an instance~$\{a_1,a_2,\dots,a_{3m}\}$ of \Partition, we construct an
instance~$(G,G_c,k,d)$ of \DCCED as follows.
The construction is illustrated in \cref{fig: DCCED NP-c}.
For graph~$G$, we first create~$m$ disjoint \emph{big} cliques each
with~$M=4(mB)^2$ vertices.
Then for every integer~$a_i$, we create a \emph{small} clique~$C_i$
with~$|C_i|=a_i$ vertices.
We set~$G_c$ to be a complete graph.
Further, we set~$k=mMB+\frac{m}{2}B^2 -\frac{1}{2}\sum_{1 \le i \le 3m}{a_i}^2$
and~$d=|E(G) \oplus E(G_c)|-k$.

Next we show that~$\{a_1,a_2,\dots,a_{3m}\}$ is a yes-instance of \Partition if
and only if~$(G,G_c,k,d)$ is a yes-instance of \DCCED.

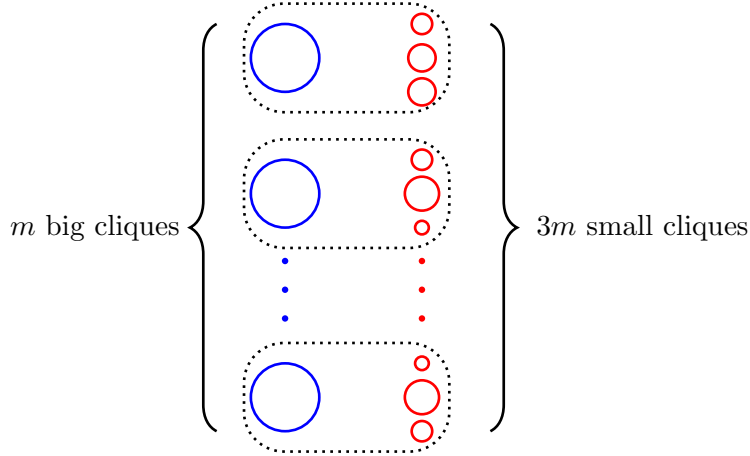
\begin{figure}[t]
\begin{center}
\begin{tikzpicture}[line width=1pt, scale=.9]

			\draw[blue] (-1,5) circle (0.5cm);
			\draw[blue] (-1,3) circle (0.5cm);
			\path (-1,2) -- (-1,1) node [blue, font=\Huge, midway, sloped] {$\dots$};
			\draw[blue] (-1,0) circle (0.5cm);
			
			\draw[red] (1,5.5) circle (0.15cm);
			\draw[red] (1,5) circle (0.2cm);
			\draw[red] (1,4.5) circle (0.2cm);
			\draw[red] (1,3.5) circle (0.15cm);
			\draw[red] (1,3) circle (0.25cm);
			\draw[red] (1,2.5) circle (0.1cm);
			\path (1,2) -- (1,1) node [red, font=\Huge, midway, sloped] {$\dots$};
			\draw[red] (1,0.5) circle (0.1cm);
			\draw[red] (1,0) circle (0.25cm);
			\draw[red] (1,-0.5) circle (0.15cm);
			
			\draw[decorate,decoration={brace,amplitude=10pt},xshift=0pt,yshift=0pt]
(-2,-0.5) -- (-2,5.5) node [black,midway,xshift=-1.6cm] 
{ $m$ big cliques};

\draw[decorate,decoration={brace,amplitude=10pt,mirror},xshift=0pt,yshift=0pt]
(2,-0.5) -- (2,5.5) node [black,midway,xshift=2cm] 
{ $3m$ small cliques};

\draw[rounded corners=15pt,dotted]  (-1.6,4.2) rectangle ++(3,1.6);
\draw[rounded corners=15pt,dotted]  (-1.6,2.2) rectangle ++(3,1.6);
\draw[rounded corners=15pt,dotted]  (-1.6,-0.8) rectangle ++(3,1.6);			

\end{tikzpicture}
\caption{Illustration of the constructed instance for the proof of \cref{lem:DCCEDhard}. 
Graph~$G$ has $m$~big cliques on the left side and~$3m$ small cliques on the right side.
Each number~$a_i$ in the instance of \Partition is represented by a small clique of size~$a_i$ on the right side.
Every dashed rounded rectangle containing one big clique and three small cliques is a possible group in a solution.}
\label{fig: DCCED NP-c}
\end{center}
\end{figure}

\emph{($\Rightarrow$):} Assume that~$\{a_1,a_2,\dots,a_{3m}\}$ is a yes-instance of \Partition. 
Then there is a partition~$A_1,A_2, \dots, A_m$ such that for each~$1 \le i \le
m$ it holds that $\sum_{a_j \in A_i}a_j=B$.
For each~$A_i$, we can combine the corresponding three small cliques and one big clique of size~$M$ into one clique.
This costs~$MB+\frac{1}{2}(B^2-\sum_{a_j \in A_i}{a_j}^2)$ edge insertions.
In total, there are
\begin{equation*}
mMB+\frac{m}{2}B^2 -\frac{1}{2}\sum_{1 \le i \le 3m}{a_i}^2=k
\end{equation*}
edge insertions.
Hence we get a cluster graph $G'$ with $|E(G) \oplus E(G')|=k$ and~$|E(G') \oplus E(G_c)|= |E(G) \oplus E(G_c)| -k =d$.

\emph{($\Leftarrow$):} Assume that~$(G,G_c,k,d)$ is a yes-instance of \DCCED and let~$G'$ be the solution.
Since~$k+d=|E(G) \oplus E(G_c)|$, to get~$G'$ we have to add exactly~$k$ edges
to~$G$.
We make the following two observations.
First, we can never combine two big cliques, as otherwise we need at least~$M^2>k$ edge insertions.
Second, every small clique must be combined with a big clique, as otherwise we have at most~$M(mB-1)$ edge insertions between big cliques and small cliques and at most~$(mB)^2$ edge insertions between small cliques, and in total there are at most~$M(mB-1)+(mB)^2=mMB-3(mB)^2<k$ edge insertions.
Hence, to get solution~$G'$ we must partition all~$3m$~small
cliques~$C_1,C_2,\dots,C_{3m}$ in~$G$ into~$m$ groups~$A_1,A_2, \dots, A_m$ and
combine all cliques in each group with one big clique.

We can split the edge insertions into two parts $k=k_1+k_2$, where $k_1=mMB$ is
the number of edge insertions between big cliques and small cliques, and
$k_2=\sum_{1 \le i \le m}{\sum_{C_j,C_k \in A_i}|C_j||C_k|}$
 is the total number of edge insertions between small cliques in each group.
We can also write $k_2$ as 
\begin{equation*}
k_2=\frac{1}{2}\sum_{1 \le i \le m}\left(\sum_{C_j \in A_i}|C_j|\right)^2-\frac{1}{2}\sum_{1 \le i \le 3m}{a_i}^2.
\end{equation*}
Recall that~$k=mMB+\frac{m}{2}B^2 -\frac{1}{2}\sum_{1 \le i \le 3m}{a_i}^2$, so
we have that~$\sum_{1 \le i \le m}({\sum_{C_j \in A_i}|C_j|)^2} = mB^2$.
Since~$\sum_{1 \le i \le 3m}|C_i|=\sum_{1 \le i \le 3m}a_i=mB$,
the equality~$\sum_{1 \le i \le m}({\sum_{C_j \in A_i}|C_j|)^2} = mB^2$ holds
only if~$C_1,C_2,\dots,C_{3m}$ can be partitioned into $m$ disjoint
subsets~$A_1,A_2, \dots, A_m$ such that for~$1 \le i \le m$ it holds that
$\sum_{C_j \in A_i}|C_j|=B$.
Thus,~$\{a_1,a_2,\dots,a_{3m}\}$ can be partitioned into~$m$ disjoint
subsets~${A_1}',{A_2}', \dots, {A_m}'$ such that for~$1 \le i \le m$ it holds
that $\sum_{a_j \in {A_i}'}a_j=B$.
\end{proof}

We continue with \DCCMD. The corresponding \NP-hardness reduction uses the same
basic ideas as in \cref{lem:DCCEDhard}.
The main difference is that in the proof of \cref{lem:DCCEDhard} we make use of
the property that we need to add exactly~$k$ edges which enforces that every
small clique should be combined with a big clique, while in the following proof
we need to make use of the matching-based distance to enforce this.

\begin{lem}
\label{lem:DCCMDhard}
\DCCMD is \NP-complete, even if the input graph~$G$ is a cluster graph.
\end{lem}

\begin{proof}
We present a polynomial-time reduction from \Partition, where given a multi-set
of~$3m$ positive integers~$\{a_1,a_2,\dots,a_{3m}\}$ with~$\sum_{1 \le i \le
3m}a_i=mB$ and for~$1 \le i \le 3m$, $B/4 < a_i < B/2$, the task is to
determine whether this multi-set can be partitioned into~$m$~disjoint
subsets~$A_1,A_2, \dots, A_m$ such that for each~$1 \le i \le m$ it holds that
$\sum_{a_j \in A_i}a_j=B$.
Given an instance~$\{a_1,a_2,\dots,a_{3m}\}$ of \Partition, we construct an
instance~$(G,G_c,k,d)$ of \DCCED as follows.

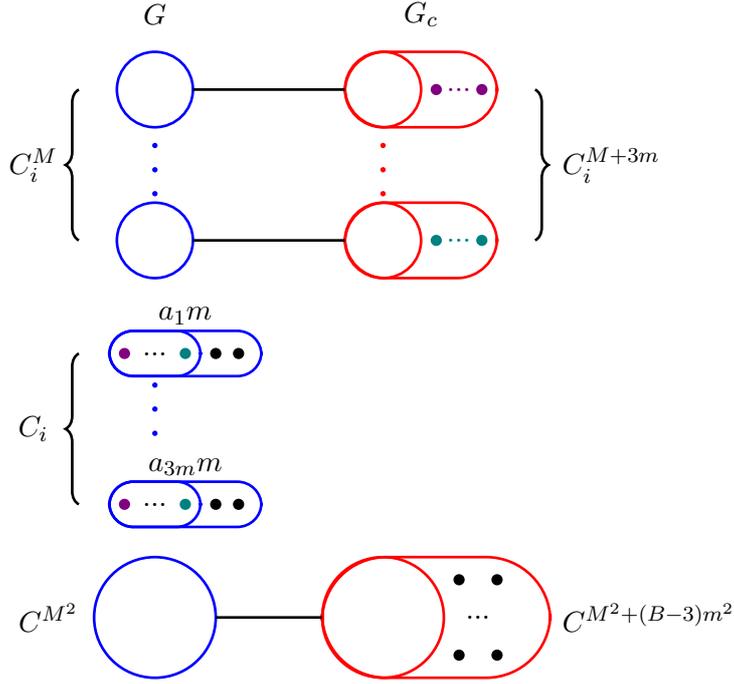
\begin{figure}[t]
\begin{center}
\begin{tikzpicture}[line width=1pt, scale=1]
		 	
		\draw[decorate,decoration={brace,amplitude=5pt},xshift=0pt,yshift=0pt]
(-2,0.5) -- (-2,2.5) node [black,midway,xshift=-0.6cm] 
{ $C_i^M$};			
			\draw[blue] (-1,2.5) circle (0.5cm);
			\path (-1,1) -- (-1,2) node [blue, font=\huge, midway, sloped] {$\dots$};
			\draw[blue] (-1,0.5) circle (0.5cm);
			
		\draw[decorate,decoration={brace,amplitude=5pt,mirror},xshift=0pt,yshift=0pt]
(4,0.5) -- (4,2.5) node [black,midway,xshift= 1cm] 
{ $C_i^{M+3m}$};
			
			\draw[red] (2,2.5) circle (0.5cm);
			\path (2,1) -- (2,2) node [red, font=\huge, midway, sloped] {$\dots$};
			\draw[red] (2,0.5) circle (0.5cm);
			\draw[red, rounded corners=15pt] (1.5,2) rectangle ++(2,1);
			\draw[red, rounded corners=15pt] (1.5,0) rectangle ++(2,1);
			
		\draw[decorate,decoration={brace,amplitude=5pt,mirror},xshift=0pt,yshift=0pt]
(-2,-1) -- (-2,-3) node [black,midway,xshift=-0.6cm] 
{ $C_i$};
			\draw[blue, rounded corners=9pt] (-1.6,-1.3) rectangle ++(1.2,0.6);
			\path (-1,-1.3) -- (-1,-2.3) node [blue, font=\huge, midway, sloped] {$\dots$};
			\draw[blue, rounded corners=9pt] (-1.6,-3.3) rectangle ++(1.2,0.6);
			\draw[blue, rounded corners=9pt] (-1.6,-1.3) rectangle ++(2,0.6);
			\draw[blue, rounded corners=9pt] (-1.6,-3.3) rectangle ++(2,0.6);
			\node (a1) at (-0.6,-0.5) {$a_1m$};
			\node (a3mm) at (-0.6,-2.5) {$a_{3m}m$};

	
\node (C^{M^2}) at (-2.4,-4.5) {$C^{M^2}$};
\node (C^{M^2}) at (5.5,-4.5) {$C^{M^2+(B-3)m^2}$};		
			\draw[blue] (-1,-4.5) circle (0.8cm);
			\draw[red] (2,-4.5) circle (0.8cm);
			\draw[red, rounded corners=24pt] (1.2,-5.3) rectangle ++(3,1.6);

			
			\node[vert,scale=0.4,violet,fill=violet] (inC1) at (-1.4,-1) {};
			\node (inC1) at (-1,-1) {...};
			\node[vert,scale=0.4,teal,fill=teal] (inC1) at (-0.6,-1) {};
			\node[vert,scale=0.4,black,fill=black] (inC1) at (-0.2,-1) {};
			\node[vert,scale=0.4,black,fill=black] (inC1) at (0.1,-1) {};
			
			\node[vert,scale=0.4,violet,fill=violet] (inC1) at (-1.4,-3) {};
			\node (inC1) at (-1,-3) {...};
			\node[vert,scale=0.4,teal,fill=teal] (inC1) at (-0.6,-3) {};
			\node[vert,scale=0.4,black,fill=black] (inC1) at (-0.2,-3) {};
			\node[vert,scale=0.4,black,fill=black] (inC1) at (0.1,-3) {};
			

			\node[vert,scale=0.4,violet,fill=violet] (inC1) at (2.7,2.5) {};
			\node[violet] (inC1) at (3,2.5) {...};
			\node[vert,scale=0.4,violet,fill=violet] (inC1) at (3.3,2.5) {};
			
			\node[vert,scale=0.4,teal,fill=teal] (inC1) at (2.7,0.5) {};
			\node[teal] (inC1) at (3,0.5) {...};
			\node[vert,scale=0.4,teal,fill=teal] (inC1) at (3.3,0.5) {};
			

			\node[vert,scale=0.4,black,fill=black] (inC1) at (3,-4) {};
			\node[vert,scale=0.4,black,fill=black] (inC1) at (3.5,-4) {};
			\node (inC1) at (3.25,-4.5) {...};
			\node[vert,scale=0.4,black,fill=black] (inC1) at (3,-5) {};
			\node[vert,scale=0.4,black,fill=black] (inC1) at (3.5,-5) {};

			\node (G) at (-1,3.5) {$G$};
			\node (Gc) at (2.5,3.5) {$G_c$};
			\draw (-0.5,2.5) -- (1.5,2.5);
			\draw (-0.5,0.5) -- (1.5,0.5);
			\draw (-0.2,-4.5) -- (1.2,-4.5);
	
\end{tikzpicture}
\caption{Illustration of the constructed instance for the proof of \cref{lem:DCCMDhard}. 
Blue borders on the left side represent cliques in graph~$G$ and red borders on the right side represent cliques in cluster graph~$G_c$.
Each number~$a_i$ in the instance of \Partition is represented by a clique~$C_i$ of size~$a_im$ in graph~$G$, which contains one vertex from every~$C_i^{M+3m}$ and~$a_i(m-1)$ vertices from~$C^{M^2+(B-3)m^2}$.
The maximum-weight matching for~$B(G,G_c)$ is formed by the edges between~$C_i^M$ and~$C_i^{M+3m}$ and the edge between~$C^{M^2}$ and~$C^{M^2+(B-3)m^2}$.}
\label{fig: DCC NP-c}
\end{center}
\end{figure}

The construction is illustrated in \cref{fig: DCC NP-c}.
For graph~$G$, we first create~$m$ big cliques~$C_1^M,C_2^M,\dots,C_m^M$ each with~$M=4(mB)^2$ vertices. 
Then for every integer~$a_i$ in~$\{a_1,a_2,\dots,a_{3m}\}$, we create a small
clique~$C_i$ with~$|C_i|=a_im$.
Lastly, we create a clique~$C^{M^2}$ with~$M^2$ vertices.
For graph~$G_c$, we create~$m+1$ cliques as follows.
For every~$C_i^M$ in~$G$, we create a clique~$C_i^{M+3m}$ with~$M+3m$ vertices
which contains all~$M$ vertices from~$C_i^M$ and one vertex from each~$C_i$ for~$1 \le i \le 3m$.
In other words, each~$C_i$ in~$G$ contains exactly one vertex from each~$C_i^{M+3m}$ in~$G_c$  for~$1 \le i \le m$.
Lastly, we create a clique~$C^{M^2+(B-3)m^2}$ which contains all remaining
vertices, that is, $M^2$~vertices from~$C^{M^2}$ and vertices from every~$C_i$ for~$1 \le i \le 3m$ which are not contained in any~$C_i^{M+3m}$.
Thus~$C^{M^2+(B-3)m^2}$ contains~$M^2+\sum_{1 \le i \le 3m}(a_i-1)m=M^2+(B-3)m^2$ vertices.
Set~$k=m^2MB+\frac{m^3}{2}B^2 -\frac{m^2}{2}\sum_{1 \le i \le 3m}{a_i}^2$ and~$d=m^2B-3m$.

It is easy to see that the maximum-weight matching~$M^*$ for~$B(G,G_c)$ is to match~$C_i^M$ with~$C_i^{M+3m}$ for every~$1 \le i \le m$ and to match~$C^{Bm^2}$ with~$C^{(2B-3)m^2}$. 
Thus the matching-based distance between~$G$ and~$G_c$ is 
\begin{equation*}
d_0=d_M(G, G_c)=\sum_{1 \le i \le 3m}a_im=m^2B.
\end{equation*}
Now, we show that~$\{a_1,a_2,\dots,a_{3m}\}$ is a yes-instance of \Partition if
and only if $(G,G_c,k,d)$ is a yes-instance of \DCCMD.

\emph{($\Rightarrow$):} Assume that~$\{a_1,a_2,\dots,a_{3m}\}$ is a yes-instance of \Partition. 
Then there is a partition~$A_i,A_2, \dots, A_m$ such that for~$1 \le i \le m$
it holds that~$\sum_{a_j \in A_i}a_j=B$.
We add edges into~$G$ to get a cluster graph~$G'$ as follows.
For each~$A_i$, we combine the corresponding three small cliques for the three
integers in~$A_i$ and the big clique~$C_i^M$ into one clique.
This costs
\begin{equation*}
MmB+m^2\sum_{a_k,a_j \in A_i}a_ja_k=MmB+\frac{m^2}{2}\left(B^2-\sum_{a_j \in
A_i}{a_j}^2\right)
\end{equation*}
edge insertions.
In total, there are~$m^2MB+\frac{m^3}{2}B^2 -\frac{m^2}{2}\sum_{1 \le i \le 3m}{a_i}^2=k$ edge insertions.
Since every small clique~$C_i$, combined with some big clique~$C_j^M$, contains one vertex from~$C_j^{M+3m}$, we obtain
\begin{equation*}
d_M(G',G_c)=d_0-3m=m^2B-3m=d.
\end{equation*}

\emph{($\Leftarrow$):} Assume that~$(G,G_c,k,d)$ is a yes-instance of \DCCMD and
let~$G'$ be the solution and let~$M'$ be the maximum-weight matching
between~$G'$ and~$G_c$.
First note that clique~$C^{M^2}$ has~$M^2$ vertices and~$M^2>k$, so we cannot combine~$C^{M^2}$ with any other clique.
Since~$M^2>(B-3)m^2$, we have also~$|C^{M^2}|>\frac{1}{2}|C^{M^2+(B-3)m^2}|$.
Hence, in the matching~$M'$ clique~$C^{M^2}$ must be matched
with~$C^{M^2+(B-3)m^2}$.
Next in the matching~$M'$ every~$C_i^{M+3m}$ in~$G_c$ must be matched with a clique in~$G'$ which contains clique~$C_i^M$, since otherwise the distance between~$G'$ and~$G_c$ is at least~$M$ and~$M>d$.
This also means that we cannot combine two big cliques~$C_i^M$ and~$C_j^M$.
Since~$d_M(G',G_c) \le d=d_0-3m$, to get solution~$G'$ every small clique~$C_i$ for~$1 \le i \le 3m$ has to be combined with some big clique~$C_j^M$.

We can split~$k$ into two parts~$k=k_1+k_2$, where~$k_1=m^2MB$ is the number of
edge insertions between big cliques and small cliques, and~$k_2$ is the total number of edge insertions between small cliques.
Similarly to the analysis in \cref{lem:DCCEDhard}, we have that~$k_2 \ge
\frac{m^3}{2}B^2 -\frac{m^2}{2}\sum_{1 \le i \le 3m}{a_i}^2$ and the equality
holds only if~$\{a_1,a_2,\dots,a_{3m}\}$ can be partitioned into~$m$ disjoint
subsets~$A_1,A_2, \dots, A_m$ such that for~$1 \le i \le m$ it holds that
$\sum_{a_j \in A_i}a_j=B$.
\end{proof}

Observe that when~$G$ is a cluster graph, then we can ``swap''~$G$ with~$G_c$ and~$k$ with~$d$:
\begin{obs}
	\label{obs: symmetric swap}
	When~$G$ is a cluster graph, instance~$(G,G_c,k,d)$ of \DCEED is a yes-instance if and only if instance~$(G_c,G,d,k)$ of \DCEED is a yes-instance.
\end{obs}

Observe that from \cref{lem:DCCEDhard,obs: symmetric swap} we can infer
\NP-hardness for \DCDED even if~$G$ is a cluster graph.
For the matching-based distance, we do not have an analogue of \cref{obs: symmetric swap}.
Thus, we provide another reduction showing \NP-hardness for \DCDMD even if~$G$ is a cluster graph.

\begin{lem}
	\label{lem:DCDMDhard}
	\DCDMD is \NP-complete, even if the input graph~$G$ is a cluster graph.
\end{lem}

\begin{proof}
We present a polynomial-time reduction from the \NP-hard \XC problem~\cite{karp1972reducibility}, where given a set~$X$ with~$|X|=3q$ and a collection~$\mathcal{S}$ of 3-element subsets of~$X$,
the task is to determine whether $\mathcal{S}$ contains a subcollection~$\mathcal{S'} \subseteq \mathcal{S}$ of size~$q$ that covers every element in~$X$ exactly once.
Given an instance~$(X,\mathcal{S})$ of \XC, where~$X=\{x_1,x_2,\dots,x_{3q}\}$ and~$\mathcal{S}=\{S_1,S_2,\dots,S_m\}$,
we construct an instance~$(G,G_c,k,d)$ of \DCDMD in polynomial time as follows.

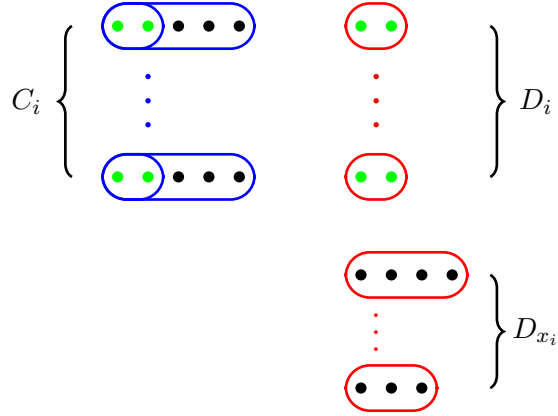
\begin{figure}[t]
\begin{center}
\begin{tikzpicture}[line width=1pt, scale=1]

		\draw[decorate,decoration={brace,amplitude=5pt,mirror},xshift=0pt,yshift=0pt]
(-2,-1) -- (-2,-3) node [black,midway,xshift=-0.6cm] 
{ $C_i$};
			\draw[blue, rounded corners=9pt] (-1.6,-1.3) rectangle ++(0.8,0.6);
			\path (-1,-1.8) -- (-1,-2.3) node [blue, font=\huge, midway, sloped] {$\dots$};
			\draw[blue, rounded corners=9pt] (-1.6,-3.3) rectangle ++(0.8,0.6);
			\draw[blue, rounded corners=9pt] (-1.6,-1.3) rectangle ++(2,0.6);
			\draw[blue, rounded corners=9pt] (-1.6,-3.3) rectangle ++(2,0.6);
		
			
			\node[vert,scale=0.4,green,fill=green] (inC1) at (-1.4,-1) {};
			\node[vert,scale=0.4,green,fill=green] (inC1) at (-1,-1) {};
			\node[vert,scale=0.4,black,fill=black] (inC1) at (-0.6,-1) {};
			\node[vert,scale=0.4,black,fill=black] (inC1) at (-0.2,-1) {};
			\node[vert,scale=0.4,black,fill=black] (inC1) at (0.2,-1) {};
			
			\node[vert,scale=0.4,green,fill=green] (inC1) at (-1.4,-3) {};
			\node[vert,scale=0.4,green,fill=green] (inC1) at (-1,-3) {};
			\node[vert,scale=0.4,black,fill=black] (inC1) at (-0.6,-3) {};
			\node[vert,scale=0.4,black,fill=black] (inC1) at (-0.2,-3) {};
			\node[vert,scale=0.4,black,fill=black] (inC1) at (0.2,-3) {};
			
			\draw[decorate,decoration={brace,amplitude=5pt},xshift=0pt,yshift=0pt]
(3.5,-1) -- (3.5,-3) node [black,midway,xshift=0.6cm] 
{ $D_i$};
			\draw[red, rounded corners=9pt] (1.6,-1.3) rectangle ++(0.8,0.6);
			\path (2,-1.8) -- (2,-2.3) node [red, font=\huge, midway, sloped] {$\dots$};
			\draw[red, rounded corners=9pt] (1.6,-3.3) rectangle ++(0.8,0.6);

\node[vert,scale=0.4,green,fill=green] (inC1) at (2.2,-1) {};
\node[vert,scale=0.4,green,fill=green] (inC1) at (1.8,-1) {};
\node[vert,scale=0.4,green,fill=green] (inC1) at (2.2,-3) {};
\node[vert,scale=0.4,green,fill=green] (inC1) at (1.8,-3) {};


\begin{scope}[shift={(0,-3.3)}]
\draw[decorate,decoration={brace,amplitude=5pt},xshift=0pt,yshift=0pt]
(3.5,-1) -- (3.5,-2.5) node [black,midway,xshift=0.6cm] 
{ $D_{x_i}$};
			\draw[red, rounded corners=9pt] (1.6,-1.3) rectangle ++(1.6,0.6);
			\path (2,-1.6) -- (2,-2) node [red, font=\Large, midway, sloped] {$\dots$};
			\draw[red, rounded corners=9pt] (1.6,-2.8) rectangle ++(1.2,0.6);

\node[vert,scale=0.4,black,fill=black] (inC1) at (1.8,-1) {};
\node[vert,scale=0.4,black,fill=black] (inC1) at (2.2,-1) {};
\node[vert,scale=0.4,black,fill=black] (inC1) at (2.6,-1) {};
\node[vert,scale=0.4,black,fill=black] (inC1) at (3,-1) {};
\node[vert,scale=0.4,black,fill=black] (inC1) at (2.6,-2.5) {};
\node[vert,scale=0.4,black,fill=black] (inC1) at (2.2,-2.5) {};
\node[vert,scale=0.4,black,fill=black] (inC1) at (1.8,-2.5) {};
\end{scope}

\end{tikzpicture}
\caption{Illustration of the constructed instance for the proof of \cref{lem:DCDMDhard}. 
Each clique~$C_i$ in~$G$ contains two green vertices, which form a clique~$D_i$ in~$G_c$.
The instance of \XC is encoded by black vertices.
On the left side each~$C_i$ in~$G$ encodes a set~$S_i$; on the right side
each~$D_{x_i}$ in~$G_c$ encodes the appearance of the element~$x_i$.}
\label{fig: DCDMD NP-c}
\end{center}
\end{figure}

The construction is illustrated in \cref{fig: DCDMD NP-c}.
For every set~$S_i=\{x_{i_1},x_{i_2},x_{i_3}\}$ in~$\mathcal{S}$, we create a
clique~$C_i=\{v_1^i,v_2^i\} \cup \{x_{i_1}^i,x_{i_2}^i,x_{i_3}^i\}$ in~$G$.
So~$G$ contains order-five~$m$ cliques~$C_1,C_2,\dots,C_m$.
For~$G_c$, we first create~$m$ cliques~$D_1,D_2,\dots,D_m$ with~$D_i=\{v_1^i,v_2^i\}$.
Then for each element~$x_i$, we create a clique~$D_{x_i}=\{x_i^j \mid x_i \in S_j\}$.
For example, if an element~$x_i$ is contained in some set~$S_j$, then in~$G$ the
corresponding clique~$C_j$ for~$S_j$ contains a vertex~$x_i^j$ which is also contained in the clique~$D_{x_i}$ in~$G_c$.
Hence, if there is a subcollection~$\mathcal{S'}$ of size~$q$ that covers every element in~$X$ exactly once, then we can find these~$q$ corresponding cliques in~$G$ and separate them to get~$3q$ new vertices each contained in one different clique~$D_{x_i}$ in~$G_c$.
Finally, we set~$k=9q$ and~$d=3m-3q$.

Note that the maximum-weight matching~$M^*$ for~$B(G,G_c)$ has to match every~$C_i$ in~$G$ with~$D_i$ in~$G_c$. 
Thus~$d_M(G, G_c)=3m$.
Now we show that~$(X,\mathcal{S})$ is a yes-instance of \XC if and only if~$(G,G_c,k,d)$ is a yes-instance of \DCDMD.

\emph{($\Rightarrow$):} Assume that~$(X,\mathcal{S})$ is a yes-instance of \XC. 
Let~$\mathcal{S'}$ be the solution.
For every~$S_i \in \mathcal{S'}$, we find the corresponding clique~$C_i=\{v_1^i,v_2^i\} \cup \{x_{i_1}^i,x_{i_2}^i,x_{i_3}^i\}$ in~$G$ and partition it into four cliques~$\{v_1^i,v_2^i\}$, $\{x_{i_1}^i\}$, $\{x_{i_2}^i\}$, and~$\{x_{i_3}^i\}$.
Let~$G'$ be the resulting cluster graph.
For every such clique~$C_i$, we delete nine edges to partition it.
Thus, overall we need to delete~$9q=k$ edges.
Since every element of~$X$ is covered by exactly one set from~$\mathcal{S'}$,
we have that in~$G'$ we get~$3q$ new cliques each with one vertex and each
vertex is contained in a different clique 
from~$D_{x_1},D_{x_2},\dots,D_{x_{3q}}$.
Thus, we have~$d_M(G,G_c)=3m-3q=d$.

\emph{($\Leftarrow$):} Assume that~$(G,G_c,k,d)$ is a yes-instance of \DCCMD.
Let~$G'$ be the solution and~$M'$ be the maximum-weight matching between~$G'$ and~$G_c$.
Since we can only delete edges to get~$G'$ and every set~$S_i$ can only contain
each element from~$X$ once, we get that in~$M'$ any edge incident on~$D_{x_i}$
has weight at most one.
Since~$d_M(G,G_c) \le d =3m-3q$, it has to be that in~$M'$ every~$D_{x_i}$ is
matched with a new clique in~$G'$ and they share exactly one vertex.
Thus we need~$3q$ new cliques in~$G'$ to be matched
with~$D_{x_1},D_{x_2},\dots,D_{x_{3q}}$ in~$G_c$.
To get these~$3q$ new cliques, we need to separate at least~$3q$ vertices
from~$C_1,C_2,\dots,C_m$ in~$G$.
Since we can delete at most~$k=9q$ edges, there have to be~$q$ cliques
from~$C_1,C_2,\dots,C_m$ such that we can separate each of them into four
parts, where the first part contains~$\{v_1^i,v_2^i\}$ and the remaining three
parts each have one vertex.
Moreover, these~$3q$ new cliques each share one vertex with one different
clique from~$D_{x_1},D_{x_2},\dots,D_{x_{3q}}$.
Thus, in the instance~$(X,\mathcal{S})$ of \XC\ we can find the
corresponding~$3q$ sets and they cover each element of~$X$ exactly once.
\end{proof}

\subsection{Parameterized Reductions}
\label{sec:Whardness}

We first show that \DCEMD is \W1-hard when parameterized by the budget~$k$.

\begin{lem}
\label{lem:DCEMDWhardwrtk}
\DCEMD is \NP-complete and \W1-hard with respect to the budget~$k$, even if the input graph~$G$ is a cluster graph.
\end{lem}

\begin{proof}
We present a parameterized reduction from \Clique, where given a graph~$G_0$ and an integer~$\ell$, we are asked to decide whether $G_0$ contains a complete subgraph of order~$\ell$. 
\Clique is \W1-hard when parameterized by~$\ell$~\cite{downey2013fundamentals}.
Given an instance~$(G_0,\ell)$ of \Clique, we construct an instance~$(G,G_c,k,d)$ of \DCEMD as follows.

The construction is illustrated in \cref{fig: DCE is W1-hard wrt k}.
Let~$n=|V(G_0)|$.
We first construct~$G$.
For every vertex~$v$ of~$G_0$, we create a clique~$C_v$ of
size~$\ell^7+\ell^4+\ell^2$.
For every edge~$e$ of~$G_0$, we create a clique~$C_e$ of size~$\ell^4+2$.
Lastly, we create a big clique~$C_B$ of size~$\ell^8$.
Note that~$G$ is already a cluster graph.
Next we construct~$G_c$.
We first create~$\ell$ cliques~$D_i$ of size~$n\ell^3$ for each~$1 \le i \le \ell$.
Every~$D_i$ contains~$\ell^3$ vertices in every~$C_v$ in~$G$.
In other words, every~$C_v$ in~$G$ contains~$\ell^3$ vertices in every~$D_i$ in~$G_c$.
Then we create a big clique~$D_B$ which contains all vertices in~$C_B$
and~$\ell^7$ vertices in every~$C_v$.
For every vertex~$v$ of~$G_0$, we create clique~$D_v$ which contains~$\ell^2$
vertices in~$C_v$ and one vertex in every~$C_e$ for~$v \in e$.
Lastly, for every edge~$e$ we create~$D_e$ which contains~$\ell^4$ vertices
in~$C_e$.
We set~$k=\binom{\ell}{2}(2\ell^4+1)+\ell\binom{\ell-1}{2}$ and we 
set~$d=d_0-\ell(\ell-1)$, where~$d_0=d_M(G,G_c)$ is the matching-based distance between~$G$ and~$G_c$, which is computed as follows.

To compute~$d_M(G,G_c)$, we need to find an optimal matching in~$B(G,G_c)$, the weighted bipartite graph between~$G$ and~$G_c$.
First, in an optimal matching~$D_B$ must be matched with~$C_B$ since~$|C_B \cap D_B|=\ell^8 > |C_v \cap D_B|=\ell^7$ for any~$v \in V(G_0)$ and~$C_B \subseteq D_B$.
Similarly, $D_e$ must be matched with~$C_e$ for every~$e \in E(G_0)$.
Then the remaining~$n$ cliques~$C_v$ in~$G$ need to be matched to~$\ell$ cliques~$D_i$ and~$n$ cliques~$D_v$ in~$G_c$.
Since~$|C_v \cap D_i|=\ell^3 > |C_v \cap D_v|=\ell^2$ for any~$v \in V(G_0)$ and~$1 \le i \le \ell$, it is always better to match~$C_v$ with some~$D_i$.
Since there are only~$\ell$ cliques~$D_i$, we can choose any~$\ell$ cliques from~$\{C_v \mid v \in V(G_0)\}$ to be matched with~$D_i$ for~$1 \le i \le \ell$ and the remaining~$n-\ell$ cliques to be matched with~$D_v$.
Thus we have many different matchings in~$B(G,G_c)$ which have the same maximum weight, and each of them corresponds to choosing~$\ell$ different cliques from~$\{C_v \mid v \in V(G_0)\}$ to be matched with~$D_i$ for~$1 \le i \le \ell$.
For each optimal matching, there are~$\ell$ \emph{free} cliques~$D_v$ in~$G_c$ which are not matched. 

\pgfdeclarelayer{bg}  
\pgfsetlayers{bg,main}

\begin{figure}[t]
\begin{center}
\begin{tikzpicture}[line width=1pt, scale=1, rotate=90]
		 	
\draw[red, dotted, rounded corners=12pt] (-.2,5.2) rectangle ++(4.2,0.8);
\draw[red, dotted, rounded corners=12pt] (-.2,3.2) rectangle ++(4.2,0.8);
\draw[red, dotted, rounded corners=12pt] (-.2,2.2) rectangle ++(4.2,0.8);
\draw[red, dotted, rounded corners=24pt] (-.2,6.2) rectangle ++(5.8,1.6);

\draw[blue, rounded corners=12pt] (0,1) rectangle ++(0.8,7);
\draw[blue, rounded corners=12pt] (1,1) rectangle ++(0.8,7);
\draw[blue, rounded corners=12pt] (3,1) rectangle ++(0.8,7);
\draw[blue] (4.8,7) circle (0.7cm);

\draw[red, dotted, rounded corners=12pt] (0,0) rectangle ++(0.8,2);
\draw[red, dotted, rounded corners=12pt] (1,0) rectangle ++(0.8,2);
\draw[red, dotted, rounded corners=12pt] (3,0) rectangle ++(0.8,2);

 \draw[blue, rounded corners=8pt] (1.6,0.2) rectangle ++(1.6,0.6);
 \draw[red, dotted, rounded corners=6pt] (1.9,0.25) rectangle ++(1,0.5);

	
	\draw[decorate,decoration={brace,amplitude=5pt,mirror},xshift=0pt,yshift=0pt]
(4.3,2.4) -- (4.3,5.8) node [midway,yshift=.5cm] 
{$D_i,1 \le i \le \ell$};

	\draw[decorate,decoration={brace,amplitude=5pt,mirror},xshift=0pt,yshift=0pt]
(3.6,8.3) -- (0.2,8.3) node [midway,xshift=-1cm] 
{$C_v,v\in V$};

\phantom{
	\draw[decorate,decoration={brace,amplitude=5pt,mirror},xshift=0pt,yshift=0pt]
(0.2,-.2) -- (3.6,-.2) node [midway,xshift=1cm] 
{$D_v,v\in V$};}

 \node () at (1.4,-.4) {$D_v$};
 \node () at (3.4,-.4) {$D_u$};
 \node () at (2.4,-1.5) {\begin{tabular}{l}$C_e$ and $D_e$\\ for $e=\{u,v\}\in E$\end{tabular}};
\node () at (4.8,7) {$C_B$};
\node () at (5.9,7) {$D_B$};

\end{tikzpicture}
\caption{Illustration of the constructed instance for the proof of \cref{lem:DCEMDWhardwrtk}.
Blue solid borders represent cliques in $G$ and red dotted borders represent cliques in $G_c$.
One horizontal long blue border represents a clique~$C_v$ in~$G$.
It has~$\ell+2$ parts and each part is contained in one clique of~$G_c$.
The first part contains~$\ell^7$~vertices which are contained in the big clique~$D_B$ of~$G_c$.
The following~$\ell$ parts each contain~$\ell^3$ vertices which are contained in the~$\ell$ cliques~$D_i$ of~$G_c$, and the last part contains $\ell^2$~vertices which are contained in~$D_v$ of~$G_c$.}
\label{fig: DCE is W1-hard wrt k}
\end{center}
\end{figure}
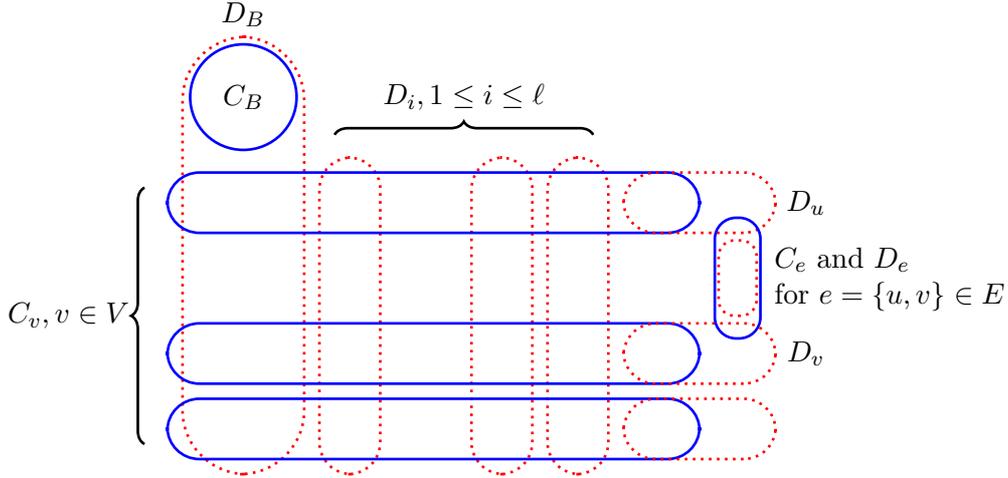

\begin{figure}[t]
\begin{center}
\begin{tikzpicture}[line width=1pt, scale=1, rotate=90]
		 	
\draw[red, dotted, rounded corners=12pt] (-.2,5.2) rectangle ++(4.2,0.8);
\draw[red, dotted, rounded corners=12pt] (-.2,3.2) rectangle ++(4.2,0.8);
\draw[red, dotted, rounded corners=12pt] (-.2,2.2) rectangle ++(4.2,0.8);
\draw[red, dotted, rounded corners=24pt] (-.2,6.2) rectangle ++(5.8,1.6);

\draw[blue, rounded corners=12pt] (0,1) rectangle ++(0.8,7);
\draw[blue, rounded corners=12pt] (1,1) rectangle ++(0.8,7);
\draw[blue, rounded corners=12pt] (3,1) rectangle ++(0.8,7);
\draw[blue,fill=green!85!black] (4.8,7) circle (0.7cm);

\draw[red, dotted, rounded corners=12pt] (0,0) rectangle ++(0.8,2);
\draw[red, dotted, rounded corners=12pt] (1,0) rectangle ++(0.8,2);
\draw[red, dotted, rounded corners=12pt] (3,0) rectangle ++(0.8,2);

 \draw[blue, rounded corners=8pt] (1.85,0.2) rectangle ++(1.1,0.6);
 \draw[red, dotted, rounded corners=6pt,fill=green!85!black] (1.9,0.25) rectangle ++(1,0.5);

	\draw[decorate,decoration={brace,amplitude=5pt,mirror},xshift=0pt,yshift=0pt]
(4.3,2.4) -- (4.3,5.8) node [midway,yshift=.5cm] 
{$D_i,1 \le i \le \ell$};

	\draw[decorate,decoration={brace,amplitude=5pt,mirror},xshift=0pt,yshift=0pt]
(3.6,8.3) -- (0.2,8.3) node [midway,xshift=-1cm] 
{$C_v,v\in V$};

\phantom{
	\draw[decorate,decoration={brace,amplitude=5pt,mirror},xshift=0pt,yshift=0pt]
(0.2,-.2) -- (3.6,-.2) node [midway,xshift=1cm] 
{$D_v,v\in V$};}

 \node () at (1.4,-.4) {$D_v$};
 \node () at (3.4,-.4) {$D_u$};
\node () at (4.8,7) {$C_B$};
\node () at (5.9,7) {$D_B$};

\draw[blue,fill=green!85!black] (3.4,0.5) circle (0.3cm);
\draw[blue,fill=green!85!black] (1.4,0.5) circle (0.3cm);

\begin{pgfonlayer}{bg}

\begin{scope}
  \clip (3,1) rectangle ++(0.8,7);
  \fill[green!85!black] (-.2,5.2) rectangle ++(4.2,0.8);
\end{scope}

\begin{scope}
  \clip (1,1) rectangle ++(0.8,7);
  \fill[green!85!black] (-.2,2.2) rectangle ++(4.2,0.8);
\end{scope}

\end{pgfonlayer}

\end{tikzpicture}
\caption{Illustration of a possible solution for the constructed instance (see \cref{fig: DCE is W1-hard wrt k}) in the proof of \cref{lem:DCEMDWhardwrtk}.
Blue solid borders represent cliques in $G'$ and red dotted borders represent cliques in $G_c$.
Green shaded areas indicate how cliques of $G'$ and $G_c$ are matched.
If two horizontal cliques of $G'$ (blue) are matched with two of the $\ell$ vertical cliques of $G_c$, then the corresponding vertices are part of the clique and hence are adjacent.
This means that the cliques corresponding to the edge can be matched in the
indicated way.
}
\label{fig: DCE is W1-hard wrt k 2}
\end{center}
\end{figure}
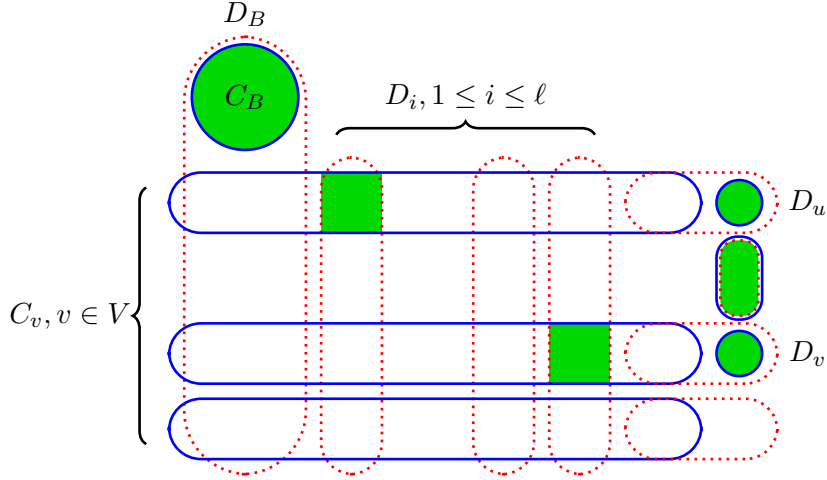
This reduction works in polynomial time.
We show that there is a clique of size~$\ell$ in~$G_0$ if and only if there is a cluster graph~$G'=(V,E')$ such that~$|E(G') \oplus E(G)| \le k$ and~$d_M(G',G_c) \le d$.

\emph{($\Rightarrow$):} Assume that there is a clique~$C^*$ of size~$\ell$ in~$G_0$. 
We modify the graph~$G$ as follows.
First, for every edge~$e$ in the clique~$C^*$ partition the corresponding clique~$C_e$ in~$G$ into three parts;
one part contains all vertices in~$D_e$ and the other two parts each have one vertex.
After this we get~$\ell(\ell-1)$ single vertices.
Since~$C^*$ is a clique, all these single vertices can be partitioned into~$\ell$ groups such that each group has~$\ell-1$ vertices and all these~$\ell-1$ vertices are contained in the same~$D_v$ for some~$v \in C^*$.
Then for each~$v \in C^*$, we combine the corresponding~$\ell-1$ vertices into one clique~$C_v^{\ell-1}$.
Denote the resulting graph as~$G'$. For an illustration see \cref{fig: DCE is W1-hard wrt k 2}.
Along the way to get~$G'$, we delete~$\binom{\ell}{2}(2\ell^4+1)$ edges and add~$\ell\binom{\ell-1}{2}$ edges, thus~$|E(G) \oplus E(G')|=\binom{\ell}{2}(2\ell^4+1)+\ell\binom{\ell-1}{2}=k$.
Next we show that~$d_M(G',G_c) \le d_0-\ell(\ell-1)$.
Recall that an optimal matching in~$B(G,G_c)$ can choose~$\ell$~cliques from~$\{C_v \mid v \in V(G_0)\}$ to be matched with~$D_i$ for~$1 \le i \le \ell$.
Now in~$B(G,G_c)$ we can choose all cliques in~$\{C_v \mid v \in C^*\}$ to be matched with~$D_i$ for~$1 \le i \le \ell$, and then match~$C_v^{\ell-1}$ with~$D_v$ for all~$v \in C^*$.
Then in the new matching we have~$\ell$ additional edges between~$C_v^{\ell-1}$ and~$D_v$ for~$v \in C^*$, each with weight~$\ell-1$.
Hence $d_M(G',G_c) \le d_0-\ell(\ell-1)$.

\emph{($\Leftarrow$):} Assume that there is a cluster graph~$G'=(V,E')$ such
that~$|E' \oplus E(G)| \le k$ and~$d_M(G',G_c) \le d$.
Note that~$k<\ell^7$, thus~$k<|C_v|$ and~$k<|C_B|$.
Consequently, we can only modify edges between vertices in~$C_e$.
It is easy to see that in any optimal matching in~$B(G',G_c)$, we still have
that clique~$C_B$ must be matched with~$D_B$ and clique~$C_e$ must be matched
with~$D_e$ for every~$e \in E(G_0)$.
We should choose~$\ell$~cliques from~$\{C_v \mid v \in V(G_0)\}$ to be
matched with~$D_i$ for~$1 \le i \le \ell$, which creates~$\ell$ free
cliques~$D_v$.
Hence, to decrease the distance between~$G$ and~$G_c$, or to increase the
matching, we have to create new cliques to be matched with these~$\ell$~free
cliques~$D_v$.
Note that every~$D_v$ only contains
single vertices from~$C_e$ with~$v \in e$ and the vertices contained
in~$C_v$. To create new cliques we need to first separate~$D_e$ to get single vertices and then combine them.
To decrease the distance by~$\ell(\ell-1)$, we need to separate at
least~$\ell(\ell-1)$ single vertices from~$C_e$.
This will cost at
least~$\ell(\ell-1)(\ell^4+1)-\binom{\ell}{2}=\binom{\ell}{2}(2\ell^4+1)$ edge
deletions if we always separate one~$C_e$ into three parts and get two single
vertices.
Then we need to combine these single vertices into at most~$\ell$ cliques since
there are at most~$\ell$ free cliques~$D_v$.
This will cost at least~$\ell\binom{\ell-1}{2}$ edge insertions if all
these~$\ell(\ell-1)$ single vertices can be partitioned into~$\ell$ groups and
each group has~$\ell-1$ vertices.
Since~$k=\binom{\ell}{2}(2\ell^4+1)+\ell\binom{\ell-1}{2}$, we have that in the
first step we have to choose~$\binom{\ell}{2}$ cliques~$C_e$ and separate them
into three parts and all these~$\ell(\ell-1)$ single vertices are evenly
distributed in~$\ell$ free cliques~$D_v$.
This means that in~$G_0$ we can select~$\binom{\ell}{2}$ edges between~$\ell$
vertices and each vertex has~$\ell-1$ incident edges.
Thus there is a clique of size~$\ell$ in~$G_0$.
\end{proof}

The next lemma shows that \DCEED is \W1-hard with respect to $k$.
The corresponding parameterized reduction is from \Clique and shares some
similarities with th reduction presented in the proof of
\cref{lem:DCEMDWhardwrtk} with respect to the edge gadgets.

The result is based on the following property for instances of \DCEED with~$k+d=|E(G) \oplus E(G_c)|$.

\begin{obs}
\label{obs: exact modification}
If an instance~$(G,G_c,k,d)$ of \DCEED (\DCDED or \DCCED) has the property
that~$k+d=|E(G) \oplus E(G_c)|$, then any solution~$G'$ satisfies that~$|E(G)
\oplus E(G')|=k$, $|E(G') \oplus E(G_c)|=d$, and~$E(G) \oplus E(G') \subseteq
E(G) \oplus E(G_c)$.
\end{obs}

\begin{proof}
On the one hand, for any graph~$G'$, we have that
\begin{equation*}
|E(G') \oplus E(G)|+|E(G') \oplus E(G_c)| \ge |E(G) \oplus E(G_c)|=k+d.
\end{equation*}
On the other hand, a solution~$G'$ satisfies that~$|E(G') \oplus E(G)| \le k$ and~$|E(G') \oplus E(G_c)| \le d$.
Thus we have that~$|E(G) \oplus E(G')|=k$ and~$|E(G') \oplus E(G_c)|=d$.

Let~$S_1=E(G) \oplus E(G') \setminus E(G) \oplus E(G_c) $ and~$S_2=E(G) \oplus E(G') \setminus S_1$.
Then
\begin{equation*}
|E(G') \oplus E(G)|=|S_1|+|S_2|
\end{equation*}
and
\begin{equation*}
|E(G') \oplus E(G_c)|=|E(G) \oplus E(G_c)|+|S_1|-|S_2|=k+d+|S_1|-|S_2|.
\end{equation*}
If~$S_1  \neq \emptyset$, then
\begin{equation*}
|E(G') \oplus E(G)|+|E(G') \oplus E(G_c)|=k+d+2|S_1|>k+d,
\end{equation*}
which is a contradiction.
Thus we conclude that $S_1=\emptyset$ and hence~$E(G) \oplus E(G') \subseteq
E(G) \oplus E(G_c)$.
\end{proof}

Consequently, when~$k+d=|E(G) \oplus E(G_c)|$ the only way to get a solution~$G'$ is to find a subset of~$E(G) \oplus E(G_c)$ with size exactly~$k$ such that modifying the edges of this subset in~$G$ yields a cluster graph.

\begin{lem}
\label{lem:DCEEDWhardwrtk}
\DCEED is \NP-complete and \W1-hard with respect to the budget~$k$, even if the input graph~$G$ is a cluster graph and~$k+d=|E(G) \oplus E(G_c)|$.
\end{lem}

\begin{proof}
We present a parameterized reduction from \Clique, where given a graph~$G_0$ and an integer~$\ell$, we are asked to decide whether $G_0$ contains a complete subgraph of order~$\ell$. 
\Clique is \W1-hard when parameterized by~$\ell$~\cite{downey2013fundamentals}.
Given an instance~$(G_0,\ell)$ of \Clique, we construct an instance~$(G,G_c,k,d)$ of \DCEED as follows.
The construction is illustrated in \cref{figure: DCEED is W1-hard wrt k exact}.
We set~$L_1=\ell^7+1$ and~$L_2=\ell^2$.
We first construct~$G$.
For every vertex~$v$ of~$G_0$, we create a clique~$C_v$ of
size~$L_1+1=\ell^7+2$, and for every edge~$e$ of~$G_0$, we create a clique~$C_e$
of size~$2L_2$.
Note that~$G$ is already a cluster graph.
Next, we construct~$G_c$.
For every vertex~$v$ of~$G_0$, let~$C_e^1,C_e^2, \dots, C_e^p$ be all cliques of size~$2L_2$ in~$G$ which represent all edges incident on~$v$.
For each vertex~$v$ of~$G_0$, we create two cliques in~$G_c$.
One of them contains~$L_1$ vertices of~$C_v$, the other contains the remaining one vertex of~$C_v$, called the \emph{single} vertex of~$C_v$, and~$L_2$~vertices from every~$C_e^i$ for~$1 \le i \le p$ (see also \cref{figure: DCEED is W1-hard wrt k exact}).
Set
\begin{equation}
\label{eq:k}
k=\ell L_1+\ell(\ell-1)L_2 + \ell \binom{\ell-1}{2}{L_2}^2+\binom{\ell}{2}{L_2}^2
\end{equation} 
and set~$d=|E(G) \oplus E(G_c)|-k$.
This reduction works in polynomial time.

\begin{figure}[t]
\begin{center}
    \begin{tikzpicture}[line width=1pt, scale=1] 

\node[draw,circle] (v) at (-2,0) {$v$};
\node (v1) at (-1.2,1.2) {};
\node[draw,circle] (v2) at (-0.5,0) {$u$};
\node (v3) at (-1.2,-1.2) {};

\draw (v) -- (v1);
\draw (v) -- (v2);
\draw (v) -- (v3);

\draw[->, double, thick] (0,0) -> (0.5,0);

\begin{scope}[scale=0.8]
\draw[red, dotted] (2,0) circle (1cm);
\draw[red, dotted] (4.5,0) circle (1.5cm);
\node[draw,circle,inner sep=1pt,fill] (v') at (4,0) {};

\draw[rounded corners=10mm, draw=blue] (0.9,1.6) -- (5,0) -- (0.9,-1.6)-- cycle;
\draw[blue] (6,0) ellipse (1cm and 0.2cm);
\draw[blue, rotate around={55:(4.5,0)}] (6,0) ellipse (1cm and 0.2cm);
\draw[blue, rotate around={-55:(4.5,0)}] (6,0) ellipse (1cm and 0.2cm);

\begin{scope}[xshift=12cm]
\begin{scope}[xscale=-1]
 \draw[red, dotted] (2,0) circle (1cm);
\draw[red, dotted] (4.5,0) circle (1.5cm);
\node[draw,circle,inner sep=1pt,fill] (v') at (4,0) {};

\draw[rounded corners=10mm, draw=blue] (0.9,1.6) -- (5,0) -- (0.9,-1.6)-- cycle;
\end{scope}
\end{scope}
\end{scope}

\node[blue] () at (2,-1.2) {$C_v$};
\node[blue] () at (7.6,-1.2) {$C_u$};
\node[blue] () at (5.5,-0.5) {$C_{v,u}$};

\end{tikzpicture}
\end{center}
\caption{Illustration of the constructed instance for the proof of \cref{lem:DCEEDWhardwrtk}.
On the left side there are two connected vertices~$v$ and~$u$ in~$G_0$ and three edges incident on~$v$.
On the right side we have the corresponding parts in~$G$ and~$G_c$.
Red dotted circles represent cliques in~$G_c$ and blue solid borders represent cliques in~$G$.}
\label{figure: DCEED is W1-hard wrt k exact}
\end{figure}
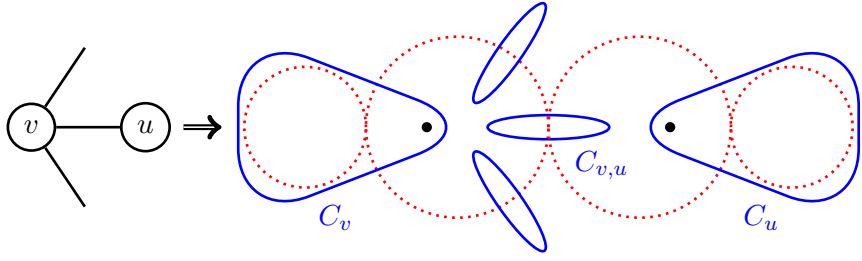

Now we show that there is a clique of size~$\ell$ in~$G_0$ if and only if there is a cluster graph~$G'=(V,E')$ such that~$|E' \oplus E(G)| \le k$ and~$|E' \oplus E(G_c)| \le d$.
To simplify the proof, we assume~$\ell \ge 3$ in the following.

\emph{($\Rightarrow$):} Assume that there is a clique~$C^*$ of size~$\ell$ in~$G_0$.
We modify graph~$G$ in the following two steps.
We first separate cliques in~$G$ according to~$C^*$ by deleting edges as follows.
For every vertex~$v$ in~$C^*$, find the clique~$C_v$ and delete edges between the single vertex in~$C_v$ and the remaining~$L_1$ vertices.
For every edge~$e$ in~$C^*$, find the clique~$C_e$ in~$G$ and delete edges to separate the clique into two parts, each with~$L_2$ vertices.
In the first step we delete~$\ell L_1+\binom{\ell}{2}{L_2}^2$ edges.
The next step is to combine some cliques by adding edges.
For every vertex~$v$ in~$C^*$, we combine the single vertex from~$C_v$ and~$\ell-1$ cliques of size~$L_2$ into one clique.
In this step we add~$\ell(\ell-1)L_2 +\ell\binom{\ell-1}{2}{L_2}^2$ edges.
Thus in total we modify~$k$~edges.

\emph{($\Leftarrow$):} Assume that there is a cluster graph~$G'=(V,E')$ such that~$|E(G') \oplus E(G)| \le k$ and~$|E(G') \oplus E(G_c)| \le d$. Since~$k+d=|E(G) \oplus E(G_c)|$, we have~$|E(G) \oplus E(G')|=k$ and~$|E(G') \oplus E(G_c)|=d$.
Thus, to get the solution~$G'$ we have to modify exactly~$k$ edges from~$E(G) \oplus E(G_c)$.
As a result, we only have the following four kinds of operations:
\begin{compactenum}
\item separate a clique~$C_v$ in~$G$ into two parts, one with the single vertex and the other with~$L_1$~vertices, which costs~$L_1=\ell^7+1$ edge deletions;
\item separate a clique~$C_e$ in~$G$ into two parts, each with~$L_2$ vertices contained in one clique in~$G_c$, which costs~${L_2}^2=\ell^4$ edge deletions;
\item combine the single vertex of~$C_v$ with some cliques of size~$L_2$ which come from separating clique~$C_e$ into two parts, which costs~$aL_2=a\ell^2$ edge insertions for some integer~$a$;
\item combine some cliques of size~$L_2$ which come from separating clique~$C_e$ into two parts, which costs~$\binom{b}{2}{L_2}^2=\binom{b}{2}\ell^4$ edge insertions for some integer~$b$.
\end{compactenum}

First, we claim that there must be~$\ell$ cliques of size~$L_1+1$ in~$G$ that have been separated.
Note that~$k=\ell^8+\frac{1}{2}\ell^7-\ell^6+\frac{1}{2}\ell^5+\ell^4-\ell^3+\ell$, where the last additive term~$\ell$ can only come from separating~$\ell$ cliques of size~$L_1+1$ in~$G$.
In addition, there cannot be more than~$\ell$ cliques of size~$L_1+1$ in~$G$
that have been separated, since~$(\ell+1) L_1 > k$ (assuming~$\ell \ge 3$).
Thus exactly~$\ell$ cliques of size~$L_1+1$ in~$G$ have to be separated and we get~$\ell$ single vertices.
This costs~$\ell L_1$ edge deletions, which is the first additive item of \cref{eq:k}.

Next, we claim that at least~$\ell(\ell-1)$ cliques of size~$L_2$ are combined
with these~$\ell$ single vertices we got in the last step.
This is because the second term of~$k$,~$\ell(\ell-1)L_2$, is strictly less than~$\ell^4$, and hence can only come from the third kind of operation, combining the single vertex with cliques of size~$L_2$.
Suppose that~$\ell(\ell-1)+\delta$ cliques of size~$L_2$ are combined with these single vertices for some~$\delta \ge 0$.
Then we need~$(\ell(\ell-1)+\delta)L_2$ edge insertions.
Note that the second additive term of \cref{eq:k} is~$\ell(\ell-1)L_2$.

Then, we need to separate at least~$\binom{\ell}{2}+\frac{\left\lceil \delta \right\rceil}{2}$ cliques of size~$2L_2$ so that we can combine them with single vertices.
Denote by~$f_1(\ell,\delta)$ the number of edge deletions this separation cost.
Clearly, $f_1(\ell,\delta) \ge (\binom{\ell}{2}+\frac{\left\lceil \delta \right\rceil}{2}){L_2}^2$.
Notice that the last additive term of \cref{eq:k} is~$\binom{\ell}{2}{L_2}^2$.

Finally, when we combine a single vertex with more than one clique of size~$L_2$, then we also need to add edges between these cliques.
Denote by~$f_2(\ell,\delta)$ the number of edge insertions between these cliques.
Since we have~$\ell(\ell-1)+\delta$ cliques of size~$L_2$ and~$\ell$ single vertices, and every clique is combined with one single vertex, it follows that~$f_2(\ell,\delta) \ge \ell \binom{\ell-1}{2}{L_2}^2$.
Notice that the third additive term of \cref{eq:k} is~$\ell \binom{\ell-1}{2}{L_2}^2$.

Overall, we need
\begin{equation*}
\ell L_1+(\ell(\ell-1)+\delta)L_2 + f_2(\ell,\delta)+f_1(\ell,\delta)
\ge k
\end{equation*}
edge modifications.
Equality only holds if~$\delta=0$, $f_1(\ell,\delta)=\binom{\ell}{2}{L_2}^2$
and~$f_2(\ell,\delta)=\ell \binom{\ell-1}{2}{L_2}^2$.
Here~$f_2(\ell,\delta)=\ell \binom{\ell-1}{2}{L_2}^2$ means that we can partition all~$\ell(\ell-1)$ cliques of size~$L$ into~$\ell$ parts, each with~$\ell-1$ cliques, and then combine all~$\ell-1$ cliques in each part with one single vertex.
Moreover,~$f_1(\ell,\delta)=\binom{\ell}{2}{L_2}^2$ means that all
these~$\ell(\ell-1)$ cliques of size~$L$ come from separating~$\binom{\ell}{2}$ cliques of size~$2L_2$.
Then, in~$G_0$ we have~$\ell$ vertices (corresponding to these~$\ell$ single vertices) and~$\binom{\ell}{2}$ edges (corresponding to these~$\binom{\ell}{2}$ cliques of size~$2L_2$) such that each vertex has~$\ell-1$ incident edges from these~$\binom{\ell}{2}$ edges.
Hence, these~$\ell$ vertices form a clique in~$G_0$.
\end{proof}

Note that in the reduction of \cref{lem:DCEEDWhardwrtk} the constructed graph~$G$ is a cluster graph.
According to \cref{obs: symmetric swap}, this reduction can also be used to prove \W1-hardness with respect to the distance~$d$.

\begin{cor}
\label{cor: DCEED is W1-hard wrt d exact}
\DCEED is \NP-complete and \W1-hard with respect to the distance~$d$, even if the input graph~$G$ is a cluster graph and~$k+d=|E(G) \oplus E(G_c)|$.
\end{cor}

The following result also exploits on the property that we need exactly~$k$ edge
modifications when~$k+d=|E(G) \oplus E(G_c)|$.

\begin{lem}
\label{lem:DCDEDWhardwrtd}
\DCDED is \W1-hard with respect to the distance~$d$, even when~$k+d=|E(G) \oplus E(G_c)|$.
\end{lem}

\begin{proof}
We present a parameterized reduction from \MCC. 
In \MCC, we are given an integer~$\ell$ and a graph where every vertex is colored with one of~$\ell$ colors.
The task is to find a clique of size~$\ell$ containing one vertex of each color.
\MCC is \W1-hard with respect to~$\ell$~\cite{fellows2009parameterized}. 
Let~$(G_0=(V,E),\ell)$ be an instance of \MCC. 
We construct an instance~$(G,G_c,k,d)$ of \DCDED as follows.
For every vertex~$v$ in~$G_0$, create a clique~$C_v$ with~$2\ell$ vertices in~$G_c$.
Add a special clique with one vertex~$v^*$ in~$G_0$.
For graph~$G$, first copy~$G_c$ and then add more edges as follows:
add edges between~$v^*$ and all other vertices in~$G$, and for every
edge~$\{u,v\}$ in~$G_0$, add all edges between vertices in~$C_u$ and vertices
in~$C_v$.
Set~$d=2\ell^2+4\ell^2 \binom{\ell}{2}$ and~$k=|E(G) \oplus E(G_c)|-d$.
This reduction works in polynomial time and the construction is illustrated in \cref{figure: DCDED is W1-hard wrt k exact}.

\begin{figure}[t]
\begin{center}
    \begin{tikzpicture}[line width=1pt, scale=1] 

\node[draw,circle,inner sep=1pt,fill] (v) at (0,0) {};
\node at (0,-0.4)  {$v^*$};
\node at (1.1,2.4)  {$C_u$};
\node at (1.8,2.0)  {$C_v$};
\draw[blue, dotted] (0,2) ellipse (1cm and 0.4cm);
\draw[blue, dotted, rotate around={60:(0,0)}] (0,2) ellipse (1cm and 0.4cm);
\draw[blue, dotted, rotate around={-60:(0,0)}] (0,2) ellipse (1cm and 0.4cm);
\draw[blue, dotted, rotate around={120:(0,0)}] (0,2) ellipse (1cm and 0.4cm);
\draw[blue, dotted, rotate around={-120:(0,0)}] (0,2) ellipse (1cm and 0.4cm);
\node (dots) at (0,-2) {$\dots$};

\node[draw,circle,scale=0.8] (v1) at (-0.8,2) {};
\node[draw,circle,scale=0.8] (v2) at (-0.4,2) {};
\node (v3) at (0.2,2) {$\dots$};
\node[draw,circle,scale=0.8] (v4) at (0.8,2) {};

\node[draw,circle,scale=0.8] (v5) at (1.4,1.6) {};

\draw[dashed] (v) -- (v1);
\draw[dashed] (v) -- (v2);
\draw[dashed] (v) -- (v4);
\draw[dashed] (v) -- (v5);
\draw[dashed] (v4) -- (v5);

\end{tikzpicture}
\end{center}
\caption{Illustration of the constructed instance for the proof of
\cref{lem:DCDEDWhardwrtd}.
A circle represents a clique in~$G_c$.
Circles in the same blue dotted ellipse mean that the corresponding vertices in~$G_0$ have the same color.
Dotted edges represent the additional edges in~$G$.
Vertex~$v^*$ is connected to all other vertices in~$G$.
For an edge~$\{u,v\}$ in~$G_0$, the corresponding two cliques~$C_u$ and~$C_v$
are connected in~$G$.}
\label{figure: DCDED is W1-hard wrt k exact}
\end{figure}
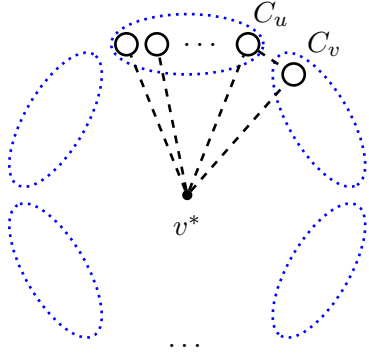

Note that~$k+d=|E(G) \oplus E(G_c)|$ and according to \cref{obs: exact modification} a solution~$G'$ for instance~$(G,G_c,k,d)$ has to delete exactly~$k$ edges from~$E(G) \oplus E(G_c)$ from~$G$, which is equivalent to adding exactly~$d$ edges from~$E(G) \oplus E(G_c)$ to~$G_c$. Next we show that there is a multicolored clique of size~$\ell$ in~$G_0$ if and only if there is a cluster graph~$G'=(V,E')$ such that~$|E' \oplus E(G)| \le k$ and~$|E' \oplus E(G_c)| \le d$.

$(\Rightarrow:)$ Suppose that there is a multicolored clique~$C_0$ of size~$\ell$ in~$G_0$, then for all vertices in~$C_0$ find the corresponding cliques in~$G_c$, and combine these~$\ell$ cliques and vertex~$v^*$ into one big clique.
Denote the resulting graph as~$G'$.
To get graph~$G'$ from~$G$, we need to delete~$|E(G) \oplus E(G_c)|-(2\ell^2+4\ell^2 \binom{\ell}{2})=k$ edges, and all these edges are in~$E(G) \oplus E(G_c)$. 
In this way we get a new cluster graph~$G'$ such that~$|E(G') \oplus E(G)| = k$ and~$|E(G') \oplus E(G_c)| = d$.

$(\Leftarrow:)$ Suppose that there is a cluster graph~$G'$ such that~$|E(G') \oplus E(G)| \le k$ and~$|E(G') \oplus E(G_c)| \le d$.
Since~$k+d=|E(G) \oplus E(G_c)|$, it has to hold that~$|E(G') \oplus E(G)| = k$ and~$|E(G') \oplus E(G_c)| = d$.
Since~$d=2\ell^2+4\ell^2 \binom{\ell}{2}$, and except for~$v^*$, every clique in~$G_c$ has~$2\ell$ vertices, so~$2\ell^2$ in~$d$ must come from adding edges between~$v^*$ and~$\ell$ cliques in~$G_c$.
Since~$G'$ is a clique, there must be edges between every pair of these~$\ell$ cliques in~$G$, which means that there is a multicolored clique of size~$\ell$ in~$G_0$.
\end{proof}

The remaining two results show \W1-hardness with respect to the distance~$d$ for \DCEMD and \DCDMD.

\begin{lem}
\label{lem:DCEMDWhardwrtd}
\DCEMD is \W1-hard with respect to the distance~$d$.
\end{lem}

\begin{proof}
We present a parameterized reduction from \CliqueReg, where given a regular graph~$G^*=(V,E)$ with vertex degree~$r$ with~$r<\frac{n}{2}$, and a number~$k^*$ with~$k^*\le r$, we are asked to decide whether $G^*$ contains a clique of size $k^*$. 
\CliqueReg is \W1-hard with respect to~$k^*$~\cite{cai2008parameterized}.

Given an instance~$(G_0, k_0, r)$ of \CliqueReg, we construct an instance~$(G, G_c, d, k)$ of \DCEMD as follows.
Graph~$G$ is the same as~$G_0$ and graph~$G_c=(V,\binom{V}{2})$ is a complete graph.
Set~$d=k_0$ and~$k=\frac{n(n-1-r)}{2}-k_0(n+k_0-2r-2)$.
The construction can trivially be done in polynomial time.
In the following we show that there is a clique of size~$k_0$ in~$G_0$ if and only if $(G, G_c, d, k)$ is a yes-instance of \DCEMD.

\emph{($\Rightarrow$):}
Assume that there is a clique of size~$k_0$ in~$G_0$; we construct a graph~$G'$ which consists of two cliques, where one of them contains the vertices from the clique of size~$k_0$ in~$G_0$; the other, denoted by~$C_\text{max}$, contains the remaining vertices and has size~$n-k_0$. 
Next we compute~$|E(G) \oplus E(G')|$, which consists of two parts: 
\begin{compactitem}
\item $D(k_0)$: the set of edges between vertices in~$C_\text{max}$ and the remaining vertices, and
\item $A(k_0)$: the set of edges between vertices in~$C_\text{max}$.
\end{compactitem}

Since the vertices outside~$C_\text{max}$ form a clique, every such vertex has~$r-k_0+1$ edges connected to vertices in~$C_\text{max}$.
Thus~$|D(k_0)|=k_0(r-k_0+1)$.
To determine~$|A(k_0)|$, we count the sum of the degrees of vertices in~$C_\text{max}$. 
Before adding edges to~$C_\text{max}$, the sum is~$(n-k_0)r$.
After adding edges the sum should be~$(n-k_0)(n-k_0-1)+|D(k_0)|$.
So the number of edges which need to added to~$C_\text{max}$ is
\begin{equation*}
|A(k_0)|=\frac{(n-k_0)(n-k_0-1)+|D(k_0)|-(n-k_0)r}{2}.
\end{equation*}

Then we get the size of the modification set for~$G'$:
\begin{equation*}
|E(G) \oplus E(G')|=|D(k_0)|+|A(k_0)|=\frac{n(n-1-r)}{2}-k_0(n+k_0-2r-2)=k.
\end{equation*}

\emph{($\Leftarrow$):}
To simplify the following proof, we define three functions:
\begin{compactitem}
\item $g_1(x):=x(r-x+1)$,
\item $g_2(x):=\frac{(n-x)(n-x-1)+g_1(x)-(n-x)r}{2}$, and
\item $f(x):=g_1(x)+g_2(x)=\frac{n(n-1-r)}{2}-x(n+x-2r-2)$.
\end{compactitem}
Since~$r <\frac{n}{2}$, we have that $f(x)$ is monotonically decreasing and~$f(k_0)=k$.

Suppose that there is no clique of size~$k_0$ in~$G_0$.
We need to show that there is no cluster graph~$G'$ satisfying both~$|E(G) \oplus E(G')| \le k$ and~$d(G_c, G') \le d$. 
Suppose towards a contradiction that there is such a cluster graph~$G'$.
Denote the largest cluster in~$G'$ as~$C_{\max}$. 
Since~$d_M(G_c, G') \le d$, we have that~$|V(C_\text{max})| \geq n-k_0$.
Define
\begin{compactitem}
\item $D$: the set of edges between vertices in~$C_\text{max}$ and the remaining vertices, and
\item $A$: the set of edges between vertices in~$C_\text{max}$.
\end{compactitem}
To get the clique~$C_{\max}$ from~$G$, we have to delete all edges in~$D$ and add all edges in~$A$, thus~$|E(G) \oplus E(G')| \ge |D| + |A|$.
We distinguish the following two cases: 

\begin{compactdesc}
\item{Case 1: $|C_\text{max}| = n-k_0$}.
Every vertex outside~$C_{\max}$ has at least~$r-k_0+1$ edges connected to vertices in~$C_{\max}$, and since there is no clique of size~$k_0$ in~$G_0$, among all vertices outside~$C_\text{max}$, there is at least one vertex which has more than~$r-k_0+1$ edges connected to vertices in~$C_\text{max}$. 
This means that~$|D|>g_1(k_0)$ and~$|A|>g_2(k_0)$. 
Thus, we have:
\begin{equation*}
  |E(G) \oplus E(G')|  \ge |D| + |A| > g_1(k_0) + g_2(k_0) = f(k_0) =k.
\end{equation*}

\item{Case 2: $|C_\text{max}| > n-k_0$}.

Suppose that $|C_\text{max}|=n-k'$, where $k'<k_0$ is the number of all vertices outside~$C_\text{max}$. 
Now we have $|D|\ge g_1(k')$ and $|A| \ge g_2(k')$, and
\begin{equation*}
|E(G) \oplus E(G')| \ge |D| + |A| \ge g_1(k') + g_2(k')= f(k') > f(k_0) = k.
\end{equation*}
The last inequality holds since~$f(k)$ is monotonically decreasing.
\end{compactdesc}
In both cases we have that there is no solution for instance~$(G, G_c, d, k)$.
\end{proof}

The above reduction cannot be used to show \W1-hardness with respect to~$d$ for
\DCDMD since both edge insertions and edge deletions are needed.
Next we show that \DCDMD remains \W1-hard with respect to~$d$.
\begin{lem}
\label{lem:DCCMDWhardwrtd}
\DCDMD{} is \W1-hard with respect to the distance~$d$.
\end{lem}

\begin{proof}
We present a parameterized reduction from \CliqueReg, where given a regular graph~$G^*=(V,E)$ with vertex degree~$r$ with~$r<\frac{n}{2}$, and number~$k^*$ with~$k^*\le r$, we are asked to decide whether $G^*$ contains a clique of size $k^*$. 
\CliqueReg is known to be \W1-hard with respect to~$k^*$~\cite{cai2008parameterized}.
Given an instance~$(G_0, \ell,r)$ of \CliqueReg, where~$G_0$ is a regular graph with vertex degree~$r$, we construct an instance~$(G, G_c, d, k)$ of \DCDMD as follows.
The construction is illustrated in \cref{figure: DCDMD is W1-hard wrt d}.

\begin{figure}[t]
\begin{center}
    \begin{tikzpicture}[line width=1pt, scale=1]

\node[draw,circle,scale=0.8] (v0) at (0,0) {};
\node[draw,circle,scale=0.8] (v1) at (-2.5,-1.5) {};
\node[draw,circle,scale=0.8] (v3) at (-0.5,-1.5) {};
\node (v2) at (-1.5,-1.5) {$\dots$};
\node[draw,circle,scale=0.8] (v4) at (0.5,-1.5) {};
\node (v5) at (1.5,-1.5) {$\dots$};
\node[draw,circle,scale=0.8] (v6) at (2.5,-1.5) {};

\node[draw,circle,scale=0.8] (v1') at (-2.5,-2.5) {};
\node[draw,circle,scale=0.8] (v3') at (-0.5,-2.5) {};
\node (v2') at (-1.5,-2.5) {$\dots$};
\node[draw,circle,scale=0.8] (v4') at (0.5,-2.5) {};
\node (v5') at (1.5,-2.5) {$\dots$};
\node[draw,circle,scale=0.8] (v6') at (2.5,-2.5) {};

\node (v_0) at (0,0.5) {$v_0$};
\node (vi) at (-3,-1.5) {$v_i$};
\node (v') at (-3,-2.5) {$v_i'$};
\node[blue] (vi) at (3.5,-1.5) {$G_0$};

\draw (v0) -- (v1);
\draw (v0) -- (v2);
\draw (v0) -- (v3);
\draw (v0) -- (v4);
\draw (v0) -- (v5);
\draw (v0) -- (v6);

\draw (v1) -- (v1');
\draw (v3) -- (v3');
\draw (v4) -- (v4');
\draw (v6) -- (v6');

\draw[rounded corners=2mm, draw=blue] (-3.5,-1.8) -- (-3.5,-1.2) -- (3,-1.2) -- (3,-1.8) -- cycle;

\end{tikzpicture}
\end{center}
\caption{
	Illustration of constructed graph~$G$ in the proof of \cref{lem:DCCMDWhardwrtd}. Edges in~$G_0$ are not shown.
	Every vertex~$v_i$ is connected to its copy~$v_i'$ and the universal
	vertex~$v_0$.
}
\label{figure: DCDMD is W1-hard wrt d}
\end{figure}
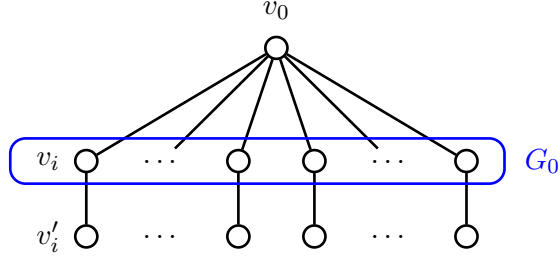

Let~$\{v_1,v_2,\dots,v_n\}$ be the vertex set of~$G_0$.
For graph~$G$, we first copy the whole graph~$G_0$.
Then we add a universal vertex and a private neighbor for each original vertex: we add a universal vertex~$v_0$ and add an edge between~$v_0$ and every vertex~$v_i$ in~$G_0$, and for every vertex~$v_i$ in~$G_0$, we add vertex~$v_i'$ and add an edge between~$v_i$ and~$v_i'$.
The graph~$G_c$ has the same vertex set as~$G$.
Moreover, graph~$G_c$ contains edges between~$v_i$ and~$v_i'$ for all~$1 \le i \le n$.
That is, $G_c$ consists of~$n+1$ cliques: $C_0$ with $V(C_0)=\{v_0\}$ and~$C_i$ with~$V(C_i)=\{v_i,v_i'\}$ for~$1 \le i \le n$.
Set~$k=n+\frac{r n}{2} -\binom{\ell}{2}$ and~$d=\ell$.
Next we show that there is clique of size~$\ell$ in~$G_0$ if and only if the constructed instance~$(G, G_c, d, k)$ is a yes-instance of \DCDMD.

\emph{($\Rightarrow$):} Assume that there is clique~$C^*$ of size~$\ell$ in~$G_0$.
Then in~$G$ we first delete edges~$\{v_i,v_i'\}$ for all~$v_i \in V(C^*)$ and
delete edges~$\{v_0,v_i\}$ for all~$v_i \in \{v_1,v_2, \dots, v_n\} \setminus
V(C^*)$.
Second, delete all edges between vertices in~$\{v_1,v_2, \dots, v_n\}$ except for edges between vertices in~$C^*$. 
We delete $n$~edges in the first step and~$\frac{r n}{2} -\binom{\ell}{2}$ edges in the second step, since~$G_0$ is a regular graph and~$C^*$ is a clique.
By deleting these~$n+\frac{r n}{2} -\binom{\ell}{2}=k$ edges, we get a cluster graph~$G'$ which contains~$n+1$ cliques:~$C_i'$ with~$V(C_i')=\{v_i'\}$ for~$v_i \in V(C^*)$,~$C_j'$ with~$V(C_j')=\{v_j,v_j'\}$ for~$v_i \in \{v_1,v_2, \dots, v_n\} \setminus V(C^*)$, and~$C_0'$ with~$V(C_0')=\{v_0\} \cup V(C^*)$.
Thus~$d_M(G',G_c)=\ell$.

\emph{($\Leftarrow$):} Assume that~$(G, G_c, d, k)$ is a yes-instance of \DCDMD and let~$G'$ be a solution.
Since we can only delete edges, for every pair of edges~$v_iv_i'$ and~$v_iv_0$
for~$v_i \in \{v_1, v_2,\dots, v_n\}$, we have to delete one of them
because~$\{v_0,v_i'\} \not\in E(G)$.
This means that for every~$v_i \in \{v_1,v_2, \dots, v_n\}$ vertex~$v_i$ is
either in the same clique with~$v_i'$ or with~$v_0$.
Suppose that in~$G'$ there are~$p \le \ell$ vertices from~$\{v_1,v_2, \dots, v_n\}$ which are in the same clique with~$v_0$.
Then these~$p$ vertices must form a clique~$C'$ in~$G_0$.
To get~$G'$, we have to delete edges~$\{v_i,v_i'\}$ for all~$v_i \in V(C')$ and
edge~$\{v_i,v_0\}$ for every vertex~$v_i \in \{v_1,v_2, \dots, v_n\} \setminus
V(C')$.
This costs~$n$~edge deletions.
Moreover, we have to delete all edges between vertices in~$\{v_1, v_2,\dots, v_n\}$ except for edges between vertices in~$C'$. 
This costs~$\frac{r n}{2} -\binom{p}{2}$ edge deletions.
Overall we have 
\begin{equation*}
|E(G) \oplus E(G')|=n+\frac{r n}{2} -\binom{p}{2}.
\end{equation*} 
Since~$G'$ is a solution, we have that~$|E(G) \oplus E(G')| \le k=n+\frac{r n}{2} -\binom{\ell}{2}$.
Hence, $p \ge \ell$ and~$G_0$ contains a clique of size~$\ell$.
\end{proof}

We now have shown all intractability results stated in \cref{thm:Whard}.

\section{Fixed-Parameter Tractability Results}\label{sec:fpt}

In this section we complement the hardness results of \cref{sec:hardness} by identifying tractable cases for the considered variants of \DCE. 
We first show that all problem variants admit a polynomial kernel for the combination of the budget~$k$ and the distance~$d$. Then we present further \FPT-results with respect to single parameters.

\subsection{Polynomial Kernels for the Combined Parameter~\boldmath$(k+d)$}\label{ssec:kernelization}

In this section we present polynomial kernels with respect to the parameter
combination~$(k+d)$ for all considered variants of \DCE:
Formally, we prove the following theorem.
\begin{thm}
\label{thm:polykernel}
The following problems admit an $O(k^2+d^2)$-vertex kernel:
\begin{compactitem}
\item \DCEMD,
\item \DCDMD, and
\item \DCCMD.
\end{compactitem}

\noindent The following problems admit an $O(k^2+k\cdot d)$-vertex kernel:
\begin{compactitem}
\item \DCEED,
\item \DCDED, and
\item \DCCED.
\end{compactitem}
All kernels can be computed in $O(|V|^3)$ time.
\end{thm}
We describe polynomial-time data reduction rules that each take an instance~$(G=(V,E),G_c=(V,E_c),k,d)$ as input and output a reduced instance.
We say that the data reduction rule is \emph{correct} if the reduced instance is a yes-instance if and only if the original instance is a yes-instance (of the corresponding problem variant).
A data reduction rule works for all problem variants that fit a given restriction.
For example, the restriction Editing/Deletion (given for \cref{rrule:heavyedge:editing}) indicates the problems: \DCEED, \DCEMD, \DCDED, and \DCDMD.
If no restriction is given, then the data reduction rules work for all problem variants.
In the correctness proof of each reduction rule, we assume that all previous rules are not applicable.

The first rule formalizes an obvious constraint on the solvability of the instance (for all problem variants). The correctness of this rule is obvious.
\rrulecounter
\begin{rrule}\label{rrule:trivial}
If $k<0$ or $d<0$, then output NO.\footnote{Formally, this does not fit the
definition of a data reduction rule, but we can assume that instead of NO the rule outputs a trivial no-instance of constant size.}
\end{rrule}

We next use some well-known reduction rules for classical \CE~\cite{GGHN05} to get a graph which consists of isolated cliques plus one vertex set of size~$k^2+2k$ that does not contain any isolated cliques. 
These rules remove edges that are part of~$k+1$ induced~$P_3$s and add edges between non-adjacent vertex pairs that are part of~$k+1$ induced~$P_3$s. 
The correctness proofs are straightforward adaptations of the correctness proofs of these rules for classical \CE.
The reason we use these data reduction rules instead of rules used for
linear-vertex kernels for classical \CE~\cite{cao_cluster_2012,chen20122k,Guo09}
is that the rules we use do not eliminate any possible solutions.
Thus, the presented rules perform edge edits that are provably part of every optimal edge modification set.

\rrulecounter
\begin{rrulev}[Editing/Deletion]\label{rrule:heavyedge:editing}
	If there are $k+1$ induced $P_3$s in $G$ that contain a common edge $\{u,v\}\in E$, then remove that edge from $E$ and decrease $k$ by one.
\end{rrulev}
\begin{rrulev}[Completion]\label{rrule:heavyedge:completion}
	If there are $k+1$ induced $P_3$s in $G$ that contain a common edge $\{u,v\}\in E$, then output NO.
\end{rrulev}

\rrulecounter
\begin{rrulev}[Editing/Completion]\label{rrule:heavynonedge:editing}
	If there are $k+1$ induced $P_3$s in $G$ that contain a common non-edge $\{u,v\}\notin E$, then add that edge to $E$ and decrease $k$ by one.
\end{rrulev}
\begin{rrulev}[Deletion]\label{rrule:heavynonedge:deletion}
	If there are $k+1$ induced $P_3$s in $G$ that contain a common non-edge $\{u,v\}\notin E$, then output NO.
\end{rrulev}
\begin{lem}
	Reduction Rules~\ref{rrule:heavyedge:editing}, \ref{rrule:heavyedge:completion}, \ref{rrule:heavynonedge:editing}, and \ref{rrule:heavynonedge:deletion} are correct.
\end{lem}
\begin{proof}
Let $I = (G=(V,E),G_c=(V,E_c),k,d)$ be an instance of a problem variant of \DCE\ and let $I^*$ be the instance after applying any of the four reduction rules. 
It is obvious that if $I$ is a no-instance, then $I^*$ is also a no-instance. 
In the following we show that if $I$ is a yes-instance, then so is $I^*$ for \cref{rrule:heavyedge:editing,rrule:heavyedge:completion}. 
The correctness of \cref{rrule:heavynonedge:editing,rrule:heavynonedge:deletion} follows by symmetric arguments. 

We now show for \cref{rrule:heavyedge:editing} that the removed edge has to be in any solution for instance $I$. 
This implies that if $I$ is a yes-instance of a completion variant of the
problem, then \cref{rrule:heavyedge:completion} is not applicable.
Assume for the sake of contradiction that $I$ is a yes-instance and that $I^*$ is a no-instance.
Then there is a cluster graph $G'$ with $|E(G')\oplus E|\le k$ that is a solution for $I$ and contains edge~$\{u,v\}$. 
However, we know that there are $k+1$ vertices $w_1, w_2, \ldots, w_{k+1}$ such that~$G[\{u,v,w_i\}]$ is a $P_3$ for all~$1\le i\le k+1$ (otherwise the rule would not be applicable). 
To destroy these~$P_3$s without removing edge~$\{u,v\}$ we need at least $k+1$ edge additions or deletions.
This is a contradiction to the assumption that $|E(G')\oplus E|\le k$.
\end{proof}

As for classical \CE\ we can upper-bound the the number of vertices that are part of $P_3$s, leading to the following reduction rule.
\rrulecounter
\begin{rrule}\label{rrule:classickernel}
	If there are more than $k^2+2k$ vertices in $V$ that are each contained in an induced~$P_3$ in $G$, then output NO.
\end{rrule}
\begin{lem}
	\cref{rrule:classickernel} is correct.
\end{lem}
\begin{proof}
Let $I = (G=(V,E),G_c=(V,E_c),k,d)$ be an instance of a problem variant of \DCE\
where \cref{rrule:heavyedge:editing} / \cref{rrule:heavyedge:completion} and
\cref{rrule:heavynonedge:editing} / \cref{rrule:heavynonedge:deletion} are not
applicable. We show that if \cref{rrule:classickernel} is applicable, then
$I$ is a no-instance.

Let $R\subseteq V$ denote
the set of vertices in~$V$ that are each contained in an induced~$P_3$ in~$G$. 
Assume for the sake of contradiction that $I$ is a yes-instance and that
\cref{rrule:classickernel} is applicable, that is, $|R|>k^2+2k$.
Then there is a cluster graph $G'$ with~$|E(G')\oplus E|\le k$ that is a solution for $I$. For each $\{u,v\}\subseteq R$, let $R_{uv}$ denote the set
of vertices $w$ such that $G[\{u,v,w\}]$ is a $P_3$. Since the aforementioned rules
are not applicable, we know that $|R_{uv}|\le k$. We further know that
$R\subseteq \bigcup_{\{u,v\}\in E(G')\oplus E} (\{u,v\}\cup R_{uv})$. It follows
that $|R|\le k(k+2)=k^2+2k$. This is a contradiction to the assumption that
\cref{rrule:classickernel} is applicable.
\end{proof}

In classical \CE\ we can just remove all isolated cliques from the graph. 
This is not always possible in our setting because of the distance constraints to $G_c$. 
However, if there is a vertex set that forms an isolated clique both in $G$ and $G_c$, then we can remove it since it has no influence on $k$ or $d$ in any problem variant. 
This is formalized in the next rule. We omit a formal correctness proof.

\rrulecounter
\begin{rrule}
\label{rrule:sameclique}
If there is a vertex set $C\subseteq V$ that is an isolated
clique in $G$ and~$G_c$, then remove all
vertices in $C$ from $G$ and $G_c$.
\end{rrule}

Now we introduce four new problem-specific reduction rules that will allow us to
upper-bound the sizes of all remaining isolated cliques and their number in a function depending on~$k+d$.
The next rules deal with large isolated cliques and allow us to either remove them or conclude that we face a no-instance.
We first state the reduction rule for the matching-based distance problem variants and then turn to the edge-based distance variants.

\rrulecounter
\begin{rrulev}[Matching-based distance]
\label{rrule:largeclique:matchingdist}
If there is a vertex set $C\subseteq V$ with $|C|>k+2d+2$ that is an isolated clique in $G$, then
\begin{compactitem}
\item if for each vertex set $C'\subseteq V$ that is an isolated clique in
$G_c$ we have that $|C\cap C'|\le d$, then answer NO,
\item otherwise, if there is a vertex set $C'\subseteq V$ that is
an isolated clique in $G_c$ and $|C\cap C'| > d$, then remove vertices
in $C$ from $G$ and $G_c$ and decrease $d$ by $|C\setminus C'|$. Furthermore,
if~$d\ge 0$, then add a set~$C_d$ of $k+d+1$ fresh vertices to~$V$. Add all
edges between vertices in $C_d$ to $G$ and add all edges between vertices in~$C_d \cup (C'\setminus C)$ to $G_c$ (if not already present).
\end{compactitem}
\begin{lem}
\cref{rrule:largeclique:matchingdist} is correct.
\end{lem}
\begin{proof}
Let $I = (G=(V,E),G_c=(V,E_c),k,d)$ be an instance of a problem variant of \DCE\ that uses the matching-based distance
and let $I^*$ be the instance after applying the reduction rule. Note that if there
is a vertex set $C\subseteq V$ with~$|C|>k+2d+2>k+1$ that is an isolated clique in
$G$, then this clique can neither be divided into smaller cliques nor can any vertex
be added to this clique, since then more than $k$ edge modifications would be
necessary (or is not allowed in the case of deletion or completion). This means that 
if $I$ is a yes-instance and $G'$ is the solution for~$I$, then for all
$\{u,v\}\in E(G')\oplus E$ we have that $\{u,v\}\cap C=\emptyset$ or, in other
words, $C$ is also an isolated clique in $G'$.

We first argue that if for each isolated clique $C'$ in $G_c$ we have that
$|C\cap C'|\le d$, then we face a no-instance.
Assume for contradiction that $I$ is a yes-instance and $G'$ is the solution for~$I$. Then we know that $C$ is also an isolated clique in $G'$ and no matter to
which clique~$C'$ in $G_c$ the clique $C$ in $G'$ is matched, we always have
that~$|C\setminus C'|>d$ and hence the matching-based distance between $G'$ and $G_c$
is too large. This is a contradiction to the assumption that $I$ is a
yes-instance.

Now assume that there is an isolated clique $C'$ in $G_c$ with $|C\cap C'| > d$. We
show that if~$I$ is a yes-instance, then~$I^*$ is a yes-instance. 
Let~$I$ be a yes-instance and let~$G'$ be a solution for~$I$.
Then we know that~$C$ is also an isolated clique in~$G'$. 
If~$C$ in~$G'$ is not matched to~$C'$ in~$G_c$ then the matching-based distance between~$G'$ and~$G_c$
is larger than~$d$. Hence, we can assume that~$C$ in~$G'$ is matched to~$C'$ in~$G_c$. 

For the next argument we introduce the following terminology. If in an
optimal solution an isolated clique~$C$ in~$G'$ is matched to an isolated clique
$C'$ in $G_c$, then we say that this match contributed $|C\setminus C'| +
|C'\setminus C|$ to the matching-based distance between~$G'$ and $G_c$.

Now we look at the
instance $I^*=(G^*=(V^*,E^*),G_c^*=(V^*,E_c^*),k^*,d^*)$ where the vertices in
$C$ are removed from $G$ and $G_c$. The reduction rule further reduces $d$ by~$|C\setminus C'|$ and introduces a new isolated
clique $C_d$ of size $k+d+1$ to $G$. In $G_c$ we have that $C_d\cup(C'\setminus C)$ is an isolated clique. 
We claim that $G^\star=(V^*,E^*\oplus (E(G')\oplus E))$ is a valid solution for
$I^*$. First, note that $G^\star$ is a cluster graph. Since the removed clique
and the added clique are both larger than $k$ we can conclude that $G^\star$ is
the graph that results from removing the clique $C$ from $G'$ and then adding
the clique $C_d$. Concerning the matching-based distance, we can replace the match
between $C$ and $C'$ in $G'$ and $G_c$, respectively, by the match between $C_d$
and $C_d\cup(C'\setminus C)$ in~$G^\star$ and~$G_c^*$, respectively. Note that the
contribution of the match between~$C$ and~$C'$ in~$G'$ and~$G_c$, respectively,
minus $|C\setminus C'|$ (the value by which $d$ is decreased by the reduction
rule) is the same as the contribution of the match between~$C_d$
and $C_d\cup(C'\setminus C)$ in~$G^\star$ and~$G_c^*$, respectively. Hence, we
can conclude that the matching-based distance between~$G'$ and~$G_c$ is the same as
the matching-based distance between~$G^\star$ and~$G_c^*$. It follows that $I^*$ is a
yes-instance.

By a symmetric argument it follows that if $I^*$ is a yes-instance, then $I$ is a
yes-instance.
\end{proof}

\end{rrulev}
\begin{rrulev}[Edge-based distance]
\label{rrule:largeclique:edgedist}
	If there is a vertex set $C\subseteq V$ with $|C|>k+1$ that is an isolated clique in $G$, then decrease $d$ by $|E_c|+\binom{|C|}{2}-2|E(G_c[C])|-|E(G_c[V\setminus C])|$ and remove vertices in~$C$ from~$G$ and $G_c$.
\end{rrulev}
\begin{lem}
\cref{rrule:largeclique:edgedist} is correct.
\end{lem}
\begin{proof}
Let $I = (G=(V,E),G_c=(V,E_c),k,d)$ be an instance of a problem variant of \DCE\ that uses the edge-based distance and let $I^*$ be the instance after applying the reduction rule. 
Note that if there is a vertex set $C\subseteq V$ with $|C|>k+1$ that is an isolated clique in $G$, then this clique can neither be divided into smaller cliques nor can any vertex be added to this clique since then more than~$k$ edge modifications would be necessary (or is not allowed in the case of deletion or completion). 
This means that if $I$ is a yes-instance and $G'$ is the solution for $I$, then for all $\{u,v\}\in E(G')\oplus E$ we have that $\{u,v\}\cap C=\emptyset$ or, in other words, $C$ is also an isolated clique in $G'$. 
This implies that removing $C$ from~$G'$ and~$G_c$ decreases $d$ by the number of edges between vertices in $C$ that are present in $G'$ but not present in~$G_c$ plus the number of edges in~$G_c$ that have one endpoint in~$C$ and one endpoint in $V\setminus G$ (note that no such edges are present in $G'$). 
The number of edges between vertices in~$C$ that are present in~$G'$ clearly is $\binom{|C|}{2}$ and the number of edges between vertices in~$C$ that are present in~$G_c$ is $|E(G_c[C])|$. 
The number of edges in~$G_c$ that have one endpoint in~$C$ and one endpoint in $V\setminus G$ is the total number of edges in~$G_c$ minus the edges in~$G_c$ between vertices in~$C$ and the edges in~$G_c$ between vertices in $V\setminus C$. 
Hence, we get~$d=\binom{|C|}{2}-|E(G_c[C])|+|E_c|-(|E(G_c[C])|+|E(G_c[V\setminus C])|)$, which yields the decrease conducted by the reduction rule. 
Note that this number is independent of~$G'$.
It follows that~$G'[V \setminus C]$ is a solution for~$I^*$.
Thus, $I^*$ is a yes-instance.

By an analogous argument we get that if $I^*$ is a yes-instance, then $I$ is also a yes-instance.
\end{proof}

If none of the previous rules are applicable, then we know that there are no large
cliques left in the graph. The next rules allow us to conclude that we face a
no-instance if there are too many small cliques left.

\rrulecounter
\begin{rrulev}[Matching-based distance]
\label{rrule:manycliques:matching}
If there are more than~$2k+d$ isolated cliques in $G$, then
output NO.
\end{rrulev}
\begin{lem}
\cref{rrule:manycliques:matching} is correct.
\end{lem}
\begin{proof}
Let $I = (G=(V,E),G_c=(V,E_c),k,d)$ be an instance of a problem variant of
\DCE\ that uses the matching-based distance. We show that if
\cref{rrule:sameclique} is not applicable and there are~$2k+d+1$ isolated cliques $C_1,C_2, \ldots, C_{2k+d+1}\subseteq V$ in~$G$, then
$I$ is a no-instance.

Assume for the sake of contradiction that $I$ is a yes-instance. Then there is a cluster
graph $G'$ with~$|E(G')\oplus E|\le k$ that is a solution for $I$. Since
\cref{rrule:sameclique} is not applicable we have that for isolated cliques
$C_i$ in $G$ with $1\le i\le 2k+d+1$ the vertex set $C_i$ is not an isolated clique in
$G_c$.
Each edge modification in $E(G')\oplus E$ when applied to $G$ can reduce the
number of isolated cliques in $G$ that are not isolated cliques in $G_c$ by at
most 2.
This happens when two isolated cliques are joined in $G$ and the union of their
vertices is an isolated clique in $G_c$. It is easy to check that this is the
best case. It follows that after~$k$~edge modifications, $G$ still has at
least $d+1$ isolated cliques that are not isolated cliques in~$G_c$. Thus the
matching-based distance cannot be decreased to $d$ which is a contradiction to the
assumption that we face a yes-instance.
\end{proof}
\begin{rrulev}[Edge-based distance]
\label{rrule:manycliques:edgedist}
If there are more than~$2(k+d)$ isolated cliques in $G$, then
output NO.
\end{rrulev}
\begin{lem}
\cref{rrule:manycliques:edgedist} is correct.
\end{lem}
\begin{proof}
Let $I = (G=(V,E),G_c=(V,E_c),k,d)$ be an instance of a problem variant of
\DCE\ that uses the edge-based distance. We show that if
\cref{rrule:sameclique} is not applicable and there are~$2(k+d)+1$ isolated cliques $C_1,C_2, \ldots, C_{2k+d+1}\subseteq V$ in~$G$, then $I$ is a no-instance.

Let $M=E \oplus E_c$. Since \cref{rrule:sameclique} is not applicable we have that for
isolated cliques $C_i$ in $G$ with $1\le i\le 2(k+d)+1$ the vertex set $C_i$ is
not an isolated clique in $G_c$. It follows that for all $C_i$ there is a vertex~$u\in
C_i$ and a vertex~$v\in V$ such that~$\{u,v\}\in M$. This implies that~$|M|>k+d$ and,
hence, we face a no-instance.
\end{proof}

In the following we show that the rules we presented decrease the number of
vertices of the instance to a number polynomial in $k+d$.

\begin{lem}
\label{lem:kernelsizematching}
Let $(G=(V,E),G_c=(V, E_c),k,d)$ be an instance of any one of the considered
problem variants of \DCE\ that uses the matching-based distance. If none of the
appropriate data reduction rules applies, then $|V|\in O(k^2 + d^2)$ and $|E|\in
O(k^3+d^3)$.
\end{lem}
\begin{proof}
Let $I = (G=(V,E),G_c=(V,E_c),k,d)$ be an instance of a problem variant of \DCE\ that uses the matching-based distance. 

Since \cref{rrule:classickernel} is not applicable we know that there are at most $k^2+2k$ vertices in~$G$ that are not part of an isolated clique. 
It is also known that there are $O(k^3)$ edges between those vertices~\cite{GGHN05}.
Further, since \cref{rrule:manycliques:matching} is not applicable, we know that there are at most $2k+d$ isolated cliques in $G$. 
Since \cref{rrule:largeclique:matchingdist} is not applicable, we know that each isolated clique has size at most $k+2d+2$. 
This yields a maximum number of $3k^2 + 2d^2 + 5dk + 2d + 6k \in O(k^2+d^2)$
vertices and $O(k^3+d^3)$ edges.
\end{proof}

\begin{lem}
\label{lem:kernelsizeedgebased}
Let $(G=(V,E), G_c=(V, E_c),k,d)$ be an instance of any one of the considered
problem variants of \DCE\ that uses the edge-based distance. If none of the
appropriate data reduction rules applies, then $|V|\in O(k^2 + k\cdot d)$ and
$|E|=O(k^3+k^2\cdot d)$.
\end{lem}
\begin{proof}
Let $I = (G=(V,E),G_c=(V,E_c),k,d)$ be an instance of a problem variant of
\DCE\ that uses the edge-based distance. 

Since \cref{rrule:classickernel} is not applicable we know that there are at most $k^2+2k$ vertices in~$G$ that are not part of an isolated clique. 
It is also known that there are $O(k^3)$ edges between those vertices~\cite{GGHN05}. 
Further, since \cref{rrule:manycliques:edgedist} is not applicable, we know that there are at most $2(k+d)$ isolated cliques in $G$.
Since \cref{rrule:largeclique:edgedist} is not applicable, we know that each isolated clique has size at most $k+1$. 
This yields a maximum number of $2 d k + 2 d + 3 k^2 + 4 k \in O(k^2+k\cdot d)$
vertices and $O(k^3+k^2\cdot d)$ edges.
\end{proof}

Finally, we can apply all data reduction rules exhaustively in~$O(|V|^3)$ time.

\begin{lem}
\label{lem:kernelcomputation}
Let $(G=(V,E),G_c=(V, E_c),k,d)$ be an instance of any one of the considered
problem variants of \DCE. Then the respective reduction rules can be
exhaustively applied in $O(|V|^3)$ time.
\end{lem}
\begin{proof}
We first exhaustively apply \cref{rrule:trivial},
\cref{rrule:heavyedge:editing}, \cref{rrule:heavyedge:completion}, \cref{rrule:heavynonedge:editing}, \cref{rrule:heavynonedge:deletion}, and
\cref{rrule:classickernel}. These rules are well-known data reduction rules for
classic \textsc{Cluster Editing} and it is known that these rules can
exhaustively be applied in $O(|V|^3)$ time~\cite{GGHN05} if the graph is
represented by an adjacency matrix.

It is easy to check that none of the remaining rules introduce new induced
$P_3$s to $G$, hence we know that once these rules (except
\cref{rrule:trivial}) are exhaustively applied, then they will not be applicable
after any of the other rules is applied.

From now on we assume that the graph $G$ is represented in the following way.
Adjacencies between vertices that are part of an induced $P_3$ are represented
in an adjacency matrix. We know that all other vertices are contained in
isolated cliques. We store a list of cliques and also a map from vertices to the
isolated clique they are contained in. We assume that $G_c$ is represented in
the same way. It is easy to check that this new representation can be computed in
$O(|V|^3)$ time from the adjacency matrix.

Using the new representation of $G$ and $G_c$, we can apply
\cref{rrule:sameclique} exhaustively in~$O(|V|^2)$~time.
We can check in~$O(|V|)$~time whether \cref{rrule:largeclique:matchingdist}
or \cref{rrule:largeclique:edgedist} is applicable and if so, apply the rule
in~$O(|V|^2)$~time. After each application of \cref{rrule:largeclique:matchingdist} we exhaustively apply
\cref{rrule:sameclique}. Since each application of \cref{rrule:largeclique:matchingdist}
or \cref{rrule:largeclique:edgedist} decreases the number of vertices in~$V$
by at least one, we can apply these rules exhaustively in~$O(|V|^3)$~time.
Using our representation of~$G$, we can apply
\cref{rrule:manycliques:matching} and \cref{rrule:manycliques:edgedist} in~$O(|V|)$~time.

Altogether, we obtain an overall running time in $O(|V|^3)$.
\end{proof}

It is easy to see that \cref{thm:polykernel} directly follows from
\cref{lem:kernelsizematching}, \cref{lem:kernelsizeedgebased}, and
\cref{lem:kernelcomputation}. We remark that the number of edges that are not
part of an isolated clique can be bounded by $O(k^3)$~\cite{GGHN05}.

\subsection{Fixed-Parameter Tractable Cases for Single Parameters}\label{ssec:general-approach}

In this section we show that several variants of \DCE\ are fixed-parameter
tractable with respect to either the budget $k$ or the distance $d$. 
\begin{thm}
	\label{thm:fpt}
	\DCDED  is in \FPT when parameterized by the budget~$k$.
	\DCCMD\ and \DCCED\ are in \FPT\ when parameterized by the distance~$d$.
\end{thm}
All our \FPT results in \cref{ssec:general-approach} are using the same
approach:
We reduce (in \FPT time) the input to an instance of \MCK (MCK), formally defined as follows.
\decprob{\MCK (MCK)}
{A family of $\ell$ mutually disjoint sets~$S_1, \ldots, S_\ell$ of items, a weight~$w_{i,j}$ and a profit~$p_{i,j}$ for each item~$j \in S_i$, and two integers~$W$ and~$P$.}
{Is it possible to select one item from each set~$S_i$ such that the total
profit is at least~$P$ and the total weight is at most~$W$?} MCK is solvable
in pseudo-polynomial time by dynamic programming:
\begin{lem}[{\cite[Section 11.5]{KPP04}}]\label{lem:MCK-pseudopoly}
	MCK can be solved in~$O(W \cdot \sum_{i=1}^\ell |S_i|)$ time.
\end{lem}
As our approach is easier to explain with the edge-based distance, we start with this case and afterwards show how to extend it to the matching-based distance.
As already exploited in our reductions showing \NP-hardness (see Theorem~\ref{thm:completionhardness}), all variants of \DCE carry some number-problem flavor.
Our generic approach will underline this flavor: 
We will focus on cases where we can partition the vertex set of the input graph into parts such that we will neither add nor delete an edge between two parts.
Moreover, we require that the parts are ``easy'' enough to list all Pareto-optimal (with respect to~$k$ and~$d$) solutions in \FPT-time (this is usually achieved by some kernelization arguments).
However, even with these strict requirements  we cannot solve the parts independently from each other:
The challenge is that we have to select for each part an appropriate Pareto-optimal solution.
Finding a feasible combination of these part-individual solutions
leads to a knapsack-type problem (in this case MCK).
Indeed, this is common to all studied variants of \DCE.


The details for our generic four-step-approach (for edge-based distance) are given subsequently.
In order to apply this approach on a concrete problem variant, we have to show (using problem specific arguments) how the requirements in the first two steps can be met.
\begin{enumerate}
	\item 
	When necessary, apply polynomial-time data reduction rules from
	\cref{ssec:kernelization}.
	Partition the input graph~$G = (V,E)$ into different parts~$G_1, G_2, \ldots,
	G_{\ell+1}$ for some $\ell\le|V|$
	such that 
	\begin{itemize}
		\item in~$G$ there is no edge between the parts and 
		\item if there is a solution, then there exists a solution where no edge between two parts will be inserted or deleted.
	\end{itemize}
	\label[step]{step:split}
	\item Compute for each part~$G_i = (V_i,E_i)$, $1\le i\le\ell$, a set~$S_i \subseteq \N^2$ encoding ``cost'' and ``gain'' of all ``representative'' solutions for~$G_i$.
		The size of the set~$S_i$ has to be upper-bounded in a function of the parameter~$p$. (Here, $p$ will be either~$k$ or~$d$.)
		
		More precisely, select a family~$\mathcal{E}_i$ of~$f(p)$ edge sets such that for each edge set~$E'_i \subseteq \binom{V_i}{2}$ in~$\mathcal{E}_i$ the graph~$G'_i = (V_i, E'_i \oplus E_i)$ is a cluster graph achievable with the allowed number of modification operations ($G'_i = (V_i, E_i \setminus E'_i)$ for edge deletions and~$G'_i = (V_i, E_i \cup E'_i)$ for edge insertions).
		For each such edge set~$E'_i$, add to~$S_i$ a tuple containing the cost ($=|E'_i|$) and ``decrease'' of the distance from~$G_i$ to the target cluster graph~$G_c$.
		More formally, for edge insertions add~$(|E'_i|, |E'_i \cap E_c| - |E'_i \setminus E_c|)$ to~$S_i$ or for edge deletions add~$(|E'_i|, |E'_i \setminus E_c|-|E'_i \cap E_c|)$ to~$S_i$, where~$E_c$ is the edge set of~$G_c$.
		Note that we allow~$E'_i = \emptyset$, that is, if~$G_i$ is a cluster graph, then~$S_i$ contains the tuple~$(0,0)$.
		
		The set~$S_i$ has to fulfill the following property:
		If there is a solution, then there is a solution~$G'$ such that restricting~$G'$ to~$V_i$ yields a tuple in~$S_i$.
		More precisely, we require that~$(|E(G'[V_i]) \oplus E_i|, |(E(G'[V_i]) \oplus E_i) \cap E_c| - |(E(G'[V_i]) \oplus E_i) \setminus E_c|) \in S_i$.%
		\label[step]{step:solutions-for-one-part}%
	\item Create an MCK instance~$I$ with~$W = k$, $P = |E \oplus E_c| - d$, and the sets~$S_1, S_2, \ldots, S_\ell$ where the tuples in the sets correspond to the items with the first number in the tuple being its weight and the second number being its profit. \label[step]{step:createMCK}
	\item Return true if and only if $I$ is a yes-instance. \label[step]{step:solveMCK}
\end{enumerate}
Note that the requirement in \cref{step:split} implies that a part is a collection of connected components in~$G$.
Furthermore, note that the part~$G_{\ell+1}$ will be ignored in the subsequent steps.
Thus~$G_{\ell+1}$ contains all vertices which are not contained in an edge of the edge modification set. 
Observe that~$\ell \le n$.
Hence, we have~$\sum_{i=1}^{\ell}|S_i| \in O(f(p)\cdot n)$. (The parameter~$p$ will be either~$k$ or~$d$.)
Moreover, as~$k$ and~$d$ are smaller than~$n^2$, it follows that~$W < n^2$ and
thus, by Lemma~\ref{lem:MCK-pseudopoly}, the MCK instance~$I$ created in
\cref{step:createMCK} can be solved in~$f(p)\cdot n^3$ time in \cref{step:solveMCK}.
This yields the following.
\begin{obs}\label{obs:fpt}
	If the partition in \cref{step:split} and the sets~$S_i$ in
	\cref{step:solutions-for-one-part} can be computed in $g(p)\cdot n^c$ time for
	some function $g$ and constant $c$ and the 
	then the above four-step-approach runs in $g(p)\cdot n^c + f(p)\cdot n^3$ time.
\end{obs}
Note that \cref{step:split} \cref{step:solutions-for-one-part} are different for every problem variant we consider.
There are, however, some similarities between the variants where only edge insertions are allowed.
 \subparagraph*{Edge-based distance.}
Next we use the above mentioned approach to show that \DCDED and \DCCED are both
fixed-para\-meter tractable with respect to~$k$ and that \DCCED is also fixed-parameter tractable with respect to~$d$.
Since \cref{step:split,step:solutions-for-one-part} are easiest to explain for
the edge-deletion variant, we start with \DCDED.
Note that the requirements of \cref{step:split,step:solutions-for-one-part}
seem impossible to achieve in FPT-time when allowing edge insertions and deletions.
Indeed, as shown in Theorem~\ref{thm:Whard}, the corresponding edge-edit
variants are \W1-hard with respect to the studied (single) parameters~$k$
and~$d$, respectively.

\begin{lem}
	\label{thm: DCDED is FPT wrt k}
	\DCDED can be solved in $O(4^{k^2}\cdot n^3)$ time and thus is in \FPT
	when parameterized by the budget~$k$.
\end{lem}

\begin{proof}
We first apply the known data reduction rules for \CE\ (see discussion after
Theorem~\ref{thm:polykernel}).
	As a result, we end up with a graph where at most~$k^2 + 2k$ vertices are contained in an induced~$P_3$; all other vertices form a cluster graph with cliques containing at most $k$ vertices each.
Denote with~$G$ the resulting graph.
	
	Now we apply our generic four-step approach. 
	Thus we need to provide the details how to implement
	\cref{step:split,step:solutions-for-one-part}. We define the parts~$G_1, G_2, \ldots, G_\ell, G_{\ell+1}$ of \cref{step:split} as follows:
	The first part~$G_1 = (V_1, E_1)$ contains the graph induced by all vertices contained in a~$P_3$.
	Each of the cliques in the cluster graph~$G[V \setminus V_1]$ forms another part~$G_i$, $2 \le i \le \ell$.
	Finally, set~$G_{\ell+1} = (\emptyset,\emptyset)$, that is, we include all vertices in the subsequent steps of our generic approach.
	Clearly, each part contains less than~$2k^2$ vertices.
	Moreover, observe that there are no edges between the parts. 
	
	As to \cref{step:solutions-for-one-part}, we add, for every edge set~$E'_i
	\subseteq E_i$ with~$G'_i = (V_i, E_i \setminus E'_i)$ being a cluster graph, a
	tuple $(|E'_i|, |E'_i \setminus E_c|-|E'_i \cap E_c|)$ to~$S_i$.
	As this enumerates all possible solutions for~$G_i$, the requirement in \cref{step:solutions-for-one-part} is fulfilled. 
	Together with \cref{obs:fpt} we get the statement of the lemma. 
\end{proof}

Next we show that \DCCED is in \FPT 
with respect to~$d$.
Before applying our generic approach, we make some observations.
Since we can only insert edges in \DCCED, we can make all vertices in a connected component pairwise adjacent.

\begin{obs}
	\label{obs: DCCED with G cluster graph}
	Let~$(G,G_c,k,d)$ be an instance of \DCCED or of \DCCMD.
	If~$G$ is not a cluster graph, then there is an equivalent instance~$(G^*,G_c,k^*,d)$ with~$G^*$ being a cluster graph and $k^* < k$.
\end{obs}
{
\begin{proof}
	Since we are only allowed to insert edges and the solution graph is required to be a cluster graph, it follows that every connected component of~$G$ has to be made into clique. 
	Call the resulting cluster graph~$G^*$ and denote by~$k'$ the number of inserted edges. 
	Clearly, $(G,G_c,k,d)$ is a yes-instance if and only if~$(G^*,G_c,k-k',d)$ is a yes-instance. 
\end{proof}
}

\cref{obs: DCCED with G cluster graph} allows us in the following to assume that all connected components in~$G$ are cliques.
The main difference between \DCDED and \DCCED is that in the former we can partition cliques independently and thus the partition for \cref{step:split} is straightforward.
In the latter problem, \cref{step:split} requires some preliminary observations.
We subsequently show that for \DCCED we can partition the graph~$G$ according to cliques in~$G_c$ such that no edges between the parts are added.

We assign every clique~$C$ in~$G$ a number that indicates the ``best match'' in~$G_c$, that is, the clique~$D$ in~$G_c$---if existing---that contains more than half of the vertices of~$C$.
More formally, let~$\mathcal{C}$ be the set all cliques in~$G$ and~$\mathcal{D}
= \{D_1, D_2, \dots, D_q\}$ be the cliques in~$G_c$.
We define a function~$T\colon \mathcal{C} \rightarrow \{0,1,\dots,q\}$ mapping a clique~$C \in \mathcal{C}$ to a number between~$0$ and~$q$ as follows:
\[
  T(C) =
  \begin{cases}
    i & \text{if } \exists i\colon |C \cap D_i| > \frac{1}{2}|C|; \\
    0 & \text{otherwise}.
  \end{cases}
\]
We say that we \emph{merge} two cliques~$C_i$ and~$C_j$ when we add all edges between~$C_i$ and~$C_j$. 
We show next that we can assume that we only merge cliques~$C_i$ and~$C_j$ with~$T(C_i) = T(C_j)$ (see \cref{figure: completion} for an illustration).

\begin{defi}
	Let~$(G,G_c,k,d)$ be a yes-instance of \DCCED such that~$G$ is a cluster graph.
	A solution~$G' = (V,E')$, $E' \supseteq E$, is called a \emph{standard} solution if for each clique~$C'$ in~$G'$ the following holds:
	If~$C_1 \neq C_2$ are two cliques in~$G$ and~$C_1, C_2 \subseteq C'$, then~$T(C_1) = T(C_2) > 0$.
\end{defi}
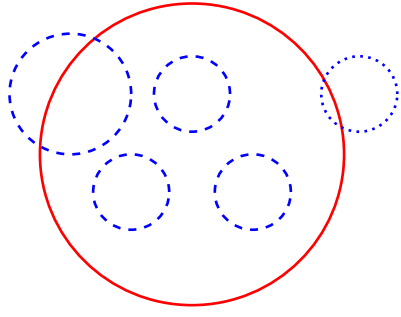
\begin{figure}[tb]
	\begin{center}
		\begin{tikzpicture}[line width=1pt, scale=1] 

			\draw[red] (0,0) circle (2cm);

			\draw[blue,dashed] (0,0.8) circle (0.5cm);
			\draw[blue,dashed] (0.8,-0.5) circle (0.5cm);
			\draw[blue,dashed] (-0.8,-0.5) circle (0.5cm);
			\draw[blue,dotted] (2.2,0.8) circle (0.5cm);
			\draw[blue,dashed] (-1.6,0.8) circle (0.8cm);
		\end{tikzpicture}
	\end{center}
	\caption{
		Illustration of a possible combination of cliques in $G$ (blue dashed or
		dotted circles).
		The red solid circle denotes a clique in $G_c$.
		Less than half of the vertices in the clique represented by the dotted (rightmost) circle are contained in the clique represented by the red solid circle.
		\cref{lem: standard solution for DCCED} states that the clique represented by the dotted circle will not be combined with a clique represented by a dashed circle.
	}
	\label{figure: completion}
\end{figure}
We start with a technical lemma that essentially states that merging two cliques~$C_i$ and~$C_j$ with~$T(C_i) \neq T(C_j)$ or~$T(C_i) = 0$ will not decrease the edge-based distance.
\begin{lem}
	\label{lem:T-map-difference}
	Let~$(G,G_c,k,d)$ be an instance of \DCCED, let~$C_0,\allowbreak C_1, \allowbreak \ldots, C_r$ be isolated cliques in~$G$ with~$r \ge 1$, and let~$E^* = \{\{u,v\} \mid u \in C_0, v \in C_1 \cup \ldots \cup C_r\}$.
	If~$T(C_0) = 0$ or if~$T(C_0) \neq T(C_i)$ for all~$1 \le i \le r$, then~$|E^* \cap E(G_c)| \le |E^* \setminus E(G_c)|$.
\end{lem}
\begin{proof}
	We first show the statement for~$r=1$.
	Let~$\mathcal{D} = \{D_1,D_2, \dots, D_q\}$ be the cliques in~$G_c$.
	Observe that~$|E^* \cap E(G_c)| = \sum_{i = 1}^{q} |C_0 \cap D_i| \cdot |C_1 \cap D_i|$.
	Furthermore, $$|E^* \setminus E(G_c)| = \sum_{i = 1}^{q} |C_0 \setminus D_i| \cdot |C_1 \cap D_i| = \sum_{i = 1}^{q} |C_0 \cap D_i| \cdot |C_1 \setminus D_i|.$$

	If~$T(C_0) = 0$, then, by definition of~$T$, we have~$|C_0 \setminus D_i| \ge
	|C_0 \cap D_i|$ for all~$1 \le i \le q$.
	Thus,~$|E^* \cap E(G_c)| \le |E^* \setminus E(G_c)|$.
	By symmetry, this follows also from~$T(C_1) = 0$.
	
	It remains to consider the case that~$T(C_0) = j_0 \ne j_1 = T(C_1)$ for
	some~$j_0,j_1 \in \{1,2, \ldots, q\}$.
	From the above argument for the case~$T(C_0) = 0$ we have
 	$$ \sum_{i \in \{1,\ldots,q\} \setminus \{j_0,j_1\}} |C_0 \cap D_i| \cdot |C_1 \cap D_i| \le \sum_{i \in \{1,\ldots,q\} \setminus \{j_0,j_1\}} |C_0 \setminus D_i| \cdot |C_1 \cap D_i|. $$ 
 	It remains to show that
 	\begin{align}
		\sum_{i \in \{j_0,j_1\}} |C_0 \cap D_i| \cdot |C_1 \cap D_i| \le \sum_{i \in \{j_0,j_1\}} |C_0 \setminus D_i| \cdot |C_1 \cap D_i|.  \label{eq:cliqueTypeEq1}
	\end{align}
 	To this end, observe that
	\begin{align}
		C_0 \cap D_{j_0} \subseteq C_0 \setminus D_{j_1} && \text{and} && C_0 \cap D_{j_1} \subseteq C_0 \setminus D_{j_0}.  \label{eq:cliqueTypeEq2}
	\end{align}
 	Furthermore, 
	\begin{align}
	 	|C_0 \setminus D_{j_1}| > |C_0 \setminus D_{j_0}| && \text{and} && |C_1 \setminus D_{j_1}| < |C_1 \setminus D_{j_0}|. \label{eq:cliqueTypeEq3}
	\end{align}
	From the above, we can deduce that
	\begin{align*}
				& |C_0 \cap D_{j_0}| \cdot |C_1 \cap D_{j_0}| + |C_0 \cap D_{j_1}| \cdot |C_1 \cap D_{j_1}| \\
		\overset{\eqref{eq:cliqueTypeEq2}}{\le} {} 	& |C_0 \setminus D_{j_1}| \cdot |C_1 \cap D_{j_0}| + |C_0 \setminus D_{j_0}| \cdot |C_1 \cap D_{j_1}| \\
		\overset{\eqref{eq:cliqueTypeEq3}}{\le} {} 	& |C_0 \setminus D_{j_0}| \cdot |C_1 \cap D_{j_0}| + |C_0 \setminus D_{j_1}| \cdot |C_1 \cap D_{j_1}|. 
 	\end{align*}
 	This completes the proof for the case~$T(C_0) \neq T(C_1)$.
 	
 	The proof for the case~$r > 1$ is now straightforward.
 	We have
 	\begin{align*}
		|E^* \cap E(G_c)| & = \sum_{j=1}^r\sum_{i = 1}^{q} |C_0 \cap D_i| \cdot |C_j \cap D_i|, \\
		|E^* \setminus E(G_c)| & = \sum_{j=1}^r\sum_{i = 1}^{q} |C_0 \setminus D_i| \cdot |C_j \cap D_i|.
 	\end{align*}
 	Using the above argument for~$r=1$ for every~$C_j$, $1 \le j \le r$, we have that
 	$$ \sum_{i = 1}^{q} |C_0 \cap D_i| \cdot |C_j \cap D_i| \le \sum_{i = 1}^{q} |C_0 \setminus D_i| \cdot |C_j \cap D_i|.$$
 	Thus, $|E^* \cap E(G_c)| \le |E^* \setminus E(G_c)|$.
\end{proof}

\begin{lem}
	\label{lem: standard solution for DCCED}
	Let~$(G,G_c,k,d)$ be a yes-instance of \DCCED.
	Then there exists a standard solution~$G'$ for~$(G,G_c,k,d)$.
\end{lem}
\begin{proof}
Let~$G'$ be a solution for $(G,G_c,k,d)$.
Assume that~$G'$ is not a \emph{standard} solution.
Let~$C'$ be a clique in~$G'$ which is not standard. (A clique~$C'$ is called a
\emph{standard} clique if either it only contains one clique
from~$\mathcal{C}$ or there is a number~$i$ with~$1 \le i \le q$ such that all
cliques from~$\mathcal{C}$ contained in~$C'$ have the value~$i$ under the
function~$T$.) Next we show that we can modify~$G'$ to get a \emph{standard}
solution by showing that each non-standard clique can be divided into some
standard cliques and the resulting cluster graph is still a solution.
We consider two cases distinguishing whether or not~$C'$ contains a clique~$C$ from~$G$ with~$T(C) = 0$.

If~$C'$ contains a clique~$C$ from~$G$ with~$T(C) = 0$, then splitting~$C$ from~$C'$ gives a cheaper solution that, by \cref{lem:T-map-difference}, also has distance at most~$d$ from~$G_c$.

Now consider the case that there is no clique~$C$ from~$G$ with~$T(C)=0$ contained in~$C'$.
Let~$S=\{C_1,C_2,\dots,C_{r'},C_{r'+1},\dots,C_r\}$ be the set of cliques from~$G$ contained in~$C'$ with~$C = C_1$.
Let~$S_1=\{C_1,C_2,\dots,C_{r'}\}$ be the set that contains all cliques~$C_i \in S$ with~$T(C_i)=T(C_1)$ and~$S_2 = S \setminus S_1$.
We divide~$C$ into two cliques~$C_{S_1}$ and~$C_{S_2}$ and get a new cluster graph~$G^*$, where~$C_{S_1}$ contains all vertices in~$C_1,C_2,\dots,C_{r'}$ and~$C_{S_2}$ contains all vertices in~$C_{r'+1},C_{r'+2},\dots,C_r$.
We show that~$G^*$ is also a solution.

Let~$E_\Delta^i$ be the set of edges between vertices in~$C_{S_1}$ and vertices in~$C_i$ for each~$r'+1 \le i \le r$.
Let~$E_\Delta=\bigcup_{r'+1 \le i \le r}E_\Delta^i$.
Then to get~$G^*$ from~$G'$, we need to delete all edges in~$E_\Delta$.
Now we split edges in $E_\Delta^i$ into two parts.
Let~$E_\Delta^i=E_+^i \uplus E_-^i$, where~$E_+^i=E_\Delta^i \cap(E(G) \oplus E(G_c))$ and~$E_-^i=E_\Delta^i \cap \big( \binom{V}{2} \setminus (E(G) \oplus E(G_c)) \big)$.
Since~$T(C_i) \neq T(C_1)$ for any~$r'+1 \le i \le r$, it follows from \cref{lem:T-map-difference} that~$|E_+^i| \le |E_-^i|$.
Then for the distance upper bound, we have that 
\begin{align*} 
|E(G^*) \oplus E(G_c)| &=  |E(G') \oplus E(G_c)| -\sum_{r'+1 \le i \le r}|E_+^i| + \sum_{r'+1 \le i \le r}|E_-^i| \\
&= |E(G') \oplus E(G_c)| -\sum_{r'+1 \le i \le r}(|E_+^i|-|E_-^i|) \\
&\le |E(G') \oplus E(G_c)| \\
&\le d.
\end{align*}
For the modification budget, we have that
\begin{equation*}
|E(G) \oplus E(G^*)| = |E(G) \oplus E(G')|-|E_\Delta| \le k.
\end{equation*}
Hence, $G^*$ is also a solution.
We can continue to divide clique~$C_{S_2}$ in~$G^*$ in the same way such that every sub-clique is standard.
Thus, for every non-standard clique in~$G'$, we can divide it into several standard cliques and the new cluster graph is still a solution. 
\end{proof}

Note that for a clique~$C_i$ in~$G$ which has exactly half of its vertices from
one clique in~$G_c$ and the remaining half of its vertices from another clique
in~$G_c$, no clique in~$G$ has the same type as it.
\cref{lem: standard solution for DCCED} allows us to focus on standard solutions.
This allows us to show our next fixed-parameter tractability result.

\begin{lem}
\label{lem: DCCED is FPT wrt d}
\DCCED can be solved in $O(d^{d+1}\cdot n^3)$ time and thus is in \FPT when
parameterized by the distance~$d$.
\end{lem}

\begin{proof}
	We apply our generic four-step approach and thus need to provide the details how to implement \cref{step:split,step:solutions-for-one-part}.
	
	By \cref{obs: DCCED with G cluster graph}, we can assume that our input graph is a cluster graph.
	Furthermore, exhaustively apply \cref{rrule:largeclique:edgedist} to delete too big cliques and denote with~$G$ the resulting cluster graph.
	Let~$\mathcal{C}$ be the set of all cliques in~$G$ and~$\mathcal{D} = \{D_1,D_2, \ldots, D_q\}$ be the set of all cliques in~$G_c$.
	We partition~$G$ into~$q+1$ groups~$G_1,G_2,\dots,G_q,G_{q+1}$ with~$G_i = G[V_i]$, where~$V_i=\{C \in \mathcal{C} \mid T(C)=i\}$ for~$1 \le i \le q$ and~$V_{q+1} = \{C \in \mathcal{C} \mid T(C)=0\}$.
	So~$G_{q+1}$ contains all cliques with value 0 under the function~$T$. 
	According to \cref{lem: standard solution for DCCED}, if there is a solution, then there is a solution only combining cliques within every group~$G_i$ for~$1 \le i \le q$.
	This shows that with~$\ell = q$ the requirements of \cref{step:split} of our generic approach are met.
		
	Next we describe \cref{step:solutions-for-one-part}, that is, for every part~$G_i$, we show how to compute a set~$S_i$ corresponding to all ``representative'' solutions.
	To this end, we distinguish two cases:~$G_i$~contains at most~$d+1$ cliques or at least~$d+2$ cliques.
	If~$G_i$ contains at most~$d+1$ cliques, then we can brute-force all possibilities to partition the cliques and merge the cliques in each partition. 
	There are less than~$(d+1)^{d+1}$ possibilities to do so and for each possibility we add to~$S_i$ a tuple representing the cost and gain of making all cliques in~$S$ into a clique.
	
	If~$G_i$ contains at least~$d+2$ cliques, then we show that we need to merge all cliques in~$G_i$:
	If not all cliques are merged into one clique, then we have a solution with two parts without any edge between the two parts (each part can be a single clique or a cluster graph).
	Let~$p$ and~$q$ be the number of vertices in the two parts that are also in~$D_i$.
	Since there are at least~$d+2$ cliques, each containing at least one vertex from~$D_i$, it follows that~$p+q \ge d+2$, $p\ge 1$, and~$q \ge 1$.
	Thus, at least~$p \cdot q \ge d+1$ edges in~$D_i$ (and thus in~$G_c$) are not in our solution, a contradiction to the fact that the solution needs to have a distance of at most~$d$ to~$G_c$.
	Hence, we only need to add one tuple to~$S_i$ encoding the cost and gain of making~$G_i$ into one clique.
	Together with \cref{obs:fpt} we get the statement of the lemma. 
\end{proof}

\subparagraph*{Matching-based distance.}

We next discuss how to adjust our generic four-step approach for \DCCMD.
The main difference to the edge-based distance variants is an additional search tree of size~$O(d^{d+2})$ in the beginning.
Each leaf of the search tree then corresponds to a simplified instance where we have additional knowledge on the matching defining the distance of a solution to~$G_c$.
With this additional knowledge, we can apply our generic four-step approach in each leaf, yielding the following.

\begin{lem}
	\label{lem: DCC is fpt wrt d}
	\DCCMD{} can be solved in $O(d^{d+2}\cdot n^3)$ time and thus is in \FPT when
parameterized by the distance~$d$.
\end{lem}
{
\begin{proof}
	We apply our generic four-step approach and thus need to provide the details how to implement \cref{step:split,step:solutions-for-one-part}.
	
	We can assume that our input graph is a cluster graph. 
	Let~$\mathcal{C}$ be the set of all cliques in~$G$ and~$\mathcal{D} = \{D_1,D_2, \ldots, D_q\}$ the set of all cliques in~$G_c$.
	Then we classify all cliques in~$\mathcal{C}$ into two classes~$\mathcal{C}_1$ and~$\mathcal{C}_2$, where every clique in~$\mathcal{C}_1$ has the property that all its vertices are contained in one clique in~$\mathcal{D}$ and every clique in~$\mathcal{C}_2$ contains vertices from at least two different cliques in~$\mathcal{D}$.
	Observe that $|\mathcal{C}_2| \le d$ as otherwise the input is a no-instance.
	Similarly, every clique in~$\mathcal{C}_2$ contains vertices from at most~$d+1$ different cliques in~$\mathcal{D}$ as otherwise the input is a no-instance.
	
	This allows us to do the following branching step. 
	For each clique in $\mathcal{C}_2$ we try out all ``meaningful'' possibilities to match it to a clique in $\mathcal{D}$, where ``meaningful'' means that the cliques in $\mathcal{C}_2$ and $\mathcal{D}$ should share some vertices or we decide to not match the clique of $\mathcal{C}_2$ to any clique in $\mathcal{D}$. 
	For each clique this gives us $d+2$ possibilities and hence we have at most~$d^{d+2}$ different cases each of which defines a mapping $M\colon \mathcal{C}_2 \rightarrow \mathcal{D}\cup\{\emptyset\}$ that maps a clique in $\mathcal{C}_2$ to the clique in $\mathcal{D}$ it is matched to.


Given the mapping $M$ from cliques in $\mathcal{C}_2$ to cliques $\mathcal{D}$ or $\emptyset$, we partition~$G$ into~$q+1$ groups~$G_1,G_2,\dots,G_q,G_{q+1}$ with~$G_i = G[V_i]$, where $V_i=\{C \in \mathcal{C}_1 \mid C\subseteq D_i\}\cup\{C \in \mathcal{C}_2 \mid M(C)=D_i\}$ and~$V_{q+1} = \{C \in \mathcal{C}_2 \mid M(C)=\emptyset \}$.

If there is a solution with a matching that uses the matches given by $M$, then
there is a solution only combining cliques within every group~$G_i$,~$1 \le i
\le q$, since all cliques in~$G_i$ that are not matched by $M$ are completely
contained in $D_i$ and hence would not be merged with cliques in~$G_j$ for some
$i\neq j$.
	This shows that with~$\ell = q$ the requirements of \cref{step:split} of our generic approach are met.

Next we describe \cref{step:solutions-for-one-part}, that is, for every
part~$G_i$, we show how to compute a set~$S_i$ corresponding to all
``representative'' solutions. Note that all except for at most $d$ cliques
from~$G_i$ need to be merged into one clique that is then matched with~$D_i$,
otherwise the matching distance would be too large. For each clique in~$G_i$
that is not completely contained in~$D_i$ we already know that it is matched
to~$D_i$, hence we need to merge all cliques of this kind to one
clique~$C^\star_i$. Each clique in $G_i$ that is completely contained in~$D_i$
and has size at least~$d+1$ also needs to be merged to~$C^\star_i$, otherwise
the matching distance would be too large. For all cliques of $G_i$ that are
completely contained in~$D_i$ with size~$x$ for some $1\le x\le d$ we merge all
but $d$ cliques to~$C^\star_i$. This leaves us with one big clique~$C^\star_i$
and $d^2$ cliques of size at most~$d$ each.
Now we can brute-force all possibilities to merge some of the remaining cliques to~$C^\star_i$.
There are less than~$d^{d}$ possibilities to do so and for each possibility we
add to~$S_i$ a tuple representing the cost and gain of merging the cliques
according to the partition.
Together with \cref{obs:fpt} we get the statement of the lemma. 
\end{proof}
}

\section{Conclusion}\label{sec:concl}
Our work provides a first thorough (parameterized) analysis of \DCE{}, addressing a natural dynamic setting for graph-based data clustering.
We deliver both (parameterized) tractability and intractability results. 
Our positive algorithmic results (fixed-parameter tractability and kernelization) are mainly of classification nature. 
Hence, to get practically useful algorithms, one needs to further improve our
running times, a challenge for future research.

A key difference between \DCE{} and static \CE{} is that all six variants of \DCE{} remain NP-hard when the input graph is a cluster graph (see Theorem~\ref{thm:completionhardness}).
Moreover, \DCE{} (both matching- and edge-based distance) is \W1-hard with
respect to the budget~$k$ (see Theorem~\ref{thm:Whard}) whereas \CE is
fixed-parameter tractable with respect to~$k$.
An obvious approach to solve \DCE{} is to compute (almost) all cluster graphs achievable with at most~$k$ edge modifications, 
then pick from this set of cluster graphs one at distance at most~$d$ to the
target cluster graph.
However, listing these cluster graphs is computationally expensive.
Indeed, our \W1-hardness results indicate that we might not do much better than using this simple approach.

We mention in passing that our results partly transfer to the ``compromise
clustering'' problem, where, given two input graphs, one wants to find a
``compromise'' cluster graph that is close enough (in terms of edge-based distance) to both input graphs.
It is easy to see that our fixed-parameter tractability results carry over if
one of these two input graphs is already a cluster graph.
A direction for future research is to examine whether our results can also be adapted to the case where both input graphs are arbitrary.
Furthermore, we left open the parameterized complexity of \DCE{} (deletion
variant and completion variant) with matching-based distance as well as \DCE{}
(completion variant) with edge-based distance when parameterized by the
budget~$k$
, see
Table~\ref{table:main-results} in Section~\ref{sec:intro}.  
Moreover, the existence of polynomial-size problem kernels for our
fixed-parameter tractable cases for single parameters (budget~$k$ or
distance~$d$) is open.

\bibliographystyle{abbrvnat} 
\bibliography{bib}

\end{document}